\documentclass[a4paper,UKenglish,cleveref, autoref, thm-restate]{lipics-v2021}

\title{Parity Games of Bounded Tree-Depth}

\author{Konrad Staniszewski}{University of Warsaw, Poland \and IDEAS NCBR Sp. z o.o., Warsaw, Poland}{}{https://orcid.org/0000-0001-7711-6827}{Work supported by the National Science Centre, Poland (grant no. 2021/41/B/ST6/03914)}

\authorrunning{K. Staniszewski} 

\Copyright{Konrad Staniszewski}



\category{}
\relatedversion{}

\acknowledgements{I want to thank the supervisor of my master's thesis Damian Niwiński and anonymous reviewers for valuable feedback.}

\hideLIPIcs\nolinenumbers

\usepackage{tikz}
\usepackage{makecell}
\usetikzlibrary{quotes,fit,shapes,arrows,positioning,patterns}
\newcommand{\WCircle}{\tikz\draw[black,fill=white] (0,0) circle (.8ex);}
\newcommand{\SCircle}{\tikz\draw[black,pattern=north west lines, pattern color=lipicsLineGray] (0,0) circle (.8ex);}
\newcommand{\GCircle}{\tikz\draw[black,fill=lipicsLightGray] (0,0) circle (.8ex);}

\newcommand{\TML}{
  \node[draw] [below left = 0.5 and 1.57 of c] (ac) {$\TDPCCVVSHORT{a}{c}$};
  \node[draw] [below left = 0.5 and 0.7 of c] (aec) {$\TDPCCVESHORT{a}{c}$};

  \node[draw] [below left = 0.5 and -0.3 of c] (bc) {$\TDPCCVVSHORT{b}{c}$};
  \node[draw] [below left = 0.5 and -1.17 of c] (bec) {$\TDPCCVESHORT{b}{c}$};

  \node[draw] [below right = 0.5 and 0.55 of c] (cc) {$\TDPCCVVSHORT{c}{c}$};
  \node[draw] [below right = 0.5 and 1.42 of c] (cec) {$\TDPCCVESHORT{c}{c}$};

  \node[draw] [below left = 0.5 and 0.6 of d2] (ad) {$\TDPCCVVSHORT{a}{d}$};
  \node[draw] [below left = 0.5 and -0.28 of d2] (aed) {$\TDPCCVESHORT{a}{d}$};
  
  \node[draw] [below right = 0.5 and -0.2 of d2] (dd) {$\TDPCCVVSHORT{d}{d}$};
  \node[draw] [below right = 0.5 and 0.68 of d2] (ded) {$\TDPCCVESHORT{d}{d}$};

  \draw (bc) -- (c);
  \draw (bec) -- (c);

  \draw (cc) -- (c);
  \draw (cec.north) -- (c);

  \draw (ac.north) -- (c);
  \draw (aec) -- (c);

  \draw (dd) -- (d2);
  \draw (ded.north) -- (d2);

  \draw (ad.north) -- (d2);
  \draw (aed) -- (d2);

}

\newcommand{\TMLI}{
  \node[draw] [below = 0.5 of b] (c) {$\TDPCCVVSHORT{c}{c}$};
  \node[draw] [below = 0.5 of d] (d2) {$\TDPCCVVSHORT{d}{d}$};
}

\newcommand{\TMLIII}{
  \node[draw] [below left = 0.5 and 1.5 of a] (b) {$\TDPCCVVSHORT{b}{b}$};
        \node[draw] [below right = 0.5 and 1.5 of a] (d) {$\TDPCCVVSHORT{d}{d}$};
}

\newcommand{\TMLIV}{
  \node[draw] [below= 0.2 of a0] (a) {$\TDPCCVVSHORT{a}{a}$};
}

\newcommand{\TMLV}[1]{
  \node[] at (0, #1) (a0) {$\ldots$};
}

\newcommand{\TMLEII}{
  \draw (c) -- (b);
  \draw (d2) -- (d);
}

\newcommand{\TMLEIV}{
  \draw (b) -- (a);
  \draw (d) -- (a);
  \draw (a) -- (a0);
}

\newcommand{\TMLEV}{
  \draw (a) -- (a0);
}

\newcommand{\CFrom}{s}
\newcommand{\CTo}{t}
\makeatletter
\newcommand{\oset}[2]{%
  \mathrel{\mathop{#2}\limits^{
    \vbox to0ex{\kern-2\ex@
    \hbox{$#1$}}}}}
\makeatother

\newcommand{\UpCloseECMP}[1]{{#1\kern-0.3em\uparrow}}

\newcommand{\MIF}{\text{if }}
\newcommand{\MOTHERWISE}{\text{otherwise }}
\newcommand{\Domain}[1]{{\mathrm{Dom}(#1)}}
\newcommand{\UpdateGadgetName}{\mathbb{U}}
\newcommand{\UpdateGadgetC}[2]{\UpdateGadgetName_{cond}^\Pair{#1}{#2}}
\newcommand{\UpdateGadgetCP}[1]{\UpdateGadgetName_{cond}^{#1}}
\newcommand{\UpdateGadget}[2]{\UpdateGadgetName^\Pair{#1}{#2}}
\newcommand{\ProcPipe}{\mathbb{P}}

\newcommand{\WOHard}{{W[1]\text{-}hard}}

\newcommand{\FT}{\rightarrow}
\newcommand{\PFT}{\rightharpoonup}
\newcommand{\TreeModelName}{\mathcal{TM}}
\newcommand{\TreeModelTupleLong}[4]{(#1, #2, #3, #4)} 
\newcommand{\TreeModelTupleShort}{\TreeModelTupleLong{\T}{\C = \C_\PE \uplus \C_\PO}{\Color}{\TCS}} 
\newcommand{\TreeModelShort}{\TreeModelName=\TreeModelTupleShort} 
\newcommand{\T}{T} 
\newcommand{\C}{C} 
\newcommand{\CPair}{\C_\PE \uplus \C_\PO}
\newcommand{\Color}{Color} 
\newcommand{\ColorBName}{COLOR} 
\newcommand{\ColorB}[2]{\ColorBName_{#2}^{#1}} 
\newcommand{\G}{G}
\newcommand{\V}{V}
\newcommand{\E}{E}
\newcommand{\R}{rank}

\newcommand{\PE}{\mathrm{E}}
\newcommand{\PO}{\mathrm{O}}
\newcommand{\PX}{\mathrm{X}}

\newcommand{\Values}{\mathrm{elem}}

\newcommand{\IFF}{\Leftrightarrow}

\newcommand{\GameGraphTuple}[2]{(#1=#1_{\PE} \uplus #1_{\PO}, #2)}
\newcommand{\GameGraphAltFullDef}[3] {#1 = \GameGraphTuple{#2}{#3}}
\newcommand{\GameGraphFullDef}[3]{#1 = \GameGraphTuple{#2}{#3 \subseteq #2 \times #2}}

\newcommand{\ParityGameFullDef}[4]{ #1=(#2=#2_{\PE} \uplus #2_{\PO}, #3 \subseteq #2 \times #2, #4 : #2 \FT \mathbb{N})}
\newcommand{\ParityGameDef}[4]{ #1=(#2=#2_{\PE} \uplus #2_{\PO}, #3, #4)}
\newcommand{\ParityGameFullDefC}[5]{ #1=(#2=#2_{\PE} \uplus #2_{\PO}, #3 \subseteq #2 \times #2, #4 : #2 \FT \mathbb{N}, #5 : #2 \FT \C_\PE \uplus \C_\PO)}
\newcommand{\ParityGameDefC}[5]{ #1=(#2=#2_{\PE} \uplus #2_{\PO}, #3, #4, #5)}

\newcommand{\MinB}[1]{{\min_{#1}}} 
\newcommand{\Min}[1]{{\min_{#1}}}
\newcommand{\Nat}{\mathbb{N}}

\newcommand{\Tuple}[1]{{\langle #1 \rangle}}
\newcommand{\Pair}[2]{\Tuple{#1, #2}}

\newcommand{\Play}{\pi}
\newcommand{\PlayPref}{\overset{\rightarrow}{\pi}}
\newcommand{\PlayPrefSpec}{\overset{\Rightarrow}{\pi}}
\newcommand{\PlayInfSpec}{\overset{\Rightarrow}{+}}
\newcommand{\PlayFrom}[3]{\Pi_{#1}(#2, #3)}
\newcommand{\PlayFromTo}[4]{\overset{\rightarrow}{\Pi}_{#1}(#2, #3, #4)}
\newcommand{\PlayLastVertex}[1]{{last(#1)}}
\newcommand{\ST}{\rho}
\newcommand{\STOP}{{\scriptstyle{STOP}}}
\newcommand{\RCMP}{\prec}
\newcommand{\RCMPEQ}{\preceq}
\newcommand{\ECMP}{\sqsubset}
\newcommand{\ECMPEQ}{\sqsubseteq}
\newcommand{\PlayMaxRank}[1]{mr_{#1}}
\newcommand{\PlayMaxRankS}[2]{mr_{#1}({#2})}

\newcommand{\FA}[2]{\forall{{#1 .}#2}}
\newcommand{\EX}[2]{\exists{{#1 .}#2}}

\newcommand{\EnforcementNameV}{\mathrm{venf}}
\newcommand{\EnforcementNameC}{\mathrm{enf}}
\newcommand{\EnforcementV}[3]{\EnforcementNameV^{#1}_{#2, #3}}

\newcommand{\EnforcementC}[4]{\EnforcementNameC^{#1}_{#3, #4}}

\newcommand{\EnfMerge}[4]{\mathrm{merge}(#2, #3, #4)}
\newcommand{\EnfSetAny}[2]{\mathcal{E}}
\newcommand{\EnfSetAlgG}[3]{\mathcal{E}^{#2}_{#3}}
\newcommand{\EnfB}[2]{\mathit{ENF}_{#2}^{#1}}
\newcommand{\EnfOneMove}{\mathit{MOVE}}
\newcommand{\EnfCycle}{\mathit{LOOPED}}
\newcommand{\EnfOpt}{OPT}
\newcommand{\EnfLoop}[2]{\mathrm{loop}(#1, #2)}
\newcommand{\NoPairWith}[2]{{#2 \not\in \Domain{#1}}}
\newcommand{\EvenPair}[2]{\mathrm{eon}(#1, #2)}

\newcommand{\EnfLift}[2]{\mathrm{lift}(#1, #2)}
\newcommand{\EnfUnion}[2]{\mathrm{\overset{\scalebox{2}[1]{$\sim$}}{\min}}(#1, #2)}

\newcommand{\UndefinedT}{undefined}
\newcommand{\UndefinedM}{\text{\UndefinedT}}

\newcommand{\IsUndefined}[1]{{#1 \text{ is \UndefinedT}}}

\newcommand{\FunAugment}[3]{{{#1}[#2 \mapsto #3]}}

\newcommand{\TCS}{D}

\newcommand{\TCSBName}{{DB}}
\newcommand{\TCSB}[2]{{\TCSBName}_{#1, #2}}
\newcommand{\TCCBName}{\mathit{CHILD}}
\newcommand{\TCCB}[2]{\TCCBName_{#1, #2}}

\newcommand{\Approaches}{\rightarrow}

\newcommand{\True}{{\mathrm{true}}}
\newcommand{\False}{{\mathrm{false}}}

\newcommand{\ACO}{{\text{AC}^0}}
\newcommand{\NCTWO}{{\text{NC}^2}}
\newcommand{\Circuit}{\mathbb{SCW}}

\newcommand{\TRCircuit}{\mathbb{RED}}
\newcommand{\TDCircuit}{\mathbb{TD}}
\newcommand{\TGCircuit}{\mathbb{GET}}
\newcommand{\IN}{\mathrm{in}}
\newcommand{\OUT}{\mathrm{out}}

\newcommand{\GUpdateEV}[2]{\mathrm{U}_{\PE}^{#1, #2}}
\newcommand{\GUpdateOV}[2]{\mathrm{U}_{\PO}^{#1, #2}}
\newcommand{\AnyGoodEnfForQPv}[3]{\mathrm{SAFE}_{#3}^{#1, #2}}
\newcommand{\AnyGoodEnfForQP}[2]{\AnyGoodEnfForQPv{#1}{#2}{}}
\newcommand{\CChoose}[2]{\mathrm{COMBINE}^{#1,#2}}
\newcommand{\CInitialize}[1]{\mathrm{INIT}_{#1}}
\newcommand{\CIn}[1]{{IN}_{#1}}
\newcommand{\COut}[1]{{OUT}^{#1}}
\newcommand{\CCond}{\mathit{COND}}
\newcommand{\CIsEColor}{\mathit{ECOLOR}}
\newcommand{\VertexUsedBit}[1]{{A}^{#1}}
\newcommand{\CLayerPack}[1]{\mathit{LPACK}_{#1}}
\newcommand{\CNoChildren}[1]{\mathrm{CHILDLESS}_{#1}}
\newcommand{\CAncestorEQ}[2]{\mathrm{ANEQ}_{#1}^{#2}}
\newcommand{\CAncestorEQEBTW}[3]{\mathrm{ANEQC}_{#1}^{#2\rightarrow#3}}

\newcommand{\Param}{k}
\newcommand{\NR}{d}
\newcommand{\ACRankBits}{(\NR+1)^\Param}
\newcommand{\ACSBitPack}{\ACRankBits}

\newcommand{\AlgComplexitySpaceL}{\mathcal{O}(\log(n) + \log(\Param))}

\newcommand{\AlgComplexitySpace}{\mathcal{O}(\log(n) + \log(\NR)\Param)}

\newcommand{\AlgComplexityP}[1]{n^{\mathcal{O}(1)}\NR^{#1}} 
\newcommand{\ACAlgNComplexity}{\AlgComplexityP{\mathcal{O}(\Param)}}
\newcommand{\ACAlgComplexity}{\ACAlgNComplexity}
\newcommand{\ACAlgDepth}{{\mathcal{O}(\Param^3)}}

\newcommand{\VRankBName}{\mathit{RANK}}
\newcommand{\VRankB}[1]{\VRankBName_{#1}}

\newcommand{\CVB}[2]{{\TDPCCVVSHORT{#1}{#2}}}
\newcommand{\CVEB}[2]{{\TDPCCVESHORT{#1}{#2}}}

\newcommand{\Forest}{\mathcal{F}}
\newcommand{\Tree}{T}

\newcommand{\EPB}[1]{\mathit{EB}[#1]}
\newcommand{\VWEdgB}[2]{\mathit{EDG}[#1, #2]}
\newcommand{\DepthB}[1]{\mathit{DEPTH}_{#1}}
\newcommand{\ParentB}[1]{\mathit{PARENT}_{#1}}
\newcommand{\CFTD}{{\Circuit_{\Tuple{2n^2, \NR + 2, 4\Param + 4}}}}

\newcommand{\TDPCCVVC}{{\circ}}
\newcommand{\TDPCCVEC}{{\rightarrow}}
\newcommand{\TDPCCVVSHORT}[2]{{{#1}_{#2}^{\TDPCCVVC{}}}}
\newcommand{\TDPCCVESHORT}[2]{{{#1}_{#2}^{\TDPCCVEC{}}}}

\newcommand{\Obrzdalek}{{Obdr\v{z}\'{a}lek}}
\newcommand{\LOGCFL}{\mathrm{LOGCFL}}
\newcommand{\PTIME}{\mathrm{P}}
\newcommand{\WinS}[1]{{W_{#1}}}

\newcommand{\FOL}{{\mathrm{FO}}}
\newcommand{\MSO}{{\mathrm{MSO}}}
\newcommand{\GSO}{{\mathrm{GSO}}}

\begin{document}

\maketitle

\begin{abstract}
The exact complexity of solving parity games is a major open problem.
Several authors have searched for efficient algorithms over specific classes of graphs.
In particular, Obdr\v{z}\'{a}lek showed that for graphs
of bounded tree-width or clique-width, the problem is in $\PTIME$, which was later improved by Ganardi, who showed that it is even in $\LOGCFL$ (with an additional assumption for clique-width case).
Here we extend this line of research by showing that for graphs of bounded tree-depth 
the problem of solving parity games is
in logspace uniform $\ACO$.
We achieve this by first considering a parameter that we obtain from a modification of clique-width, which we call shallow clique-width.
We subsequently provide a suitable reduction.

\end{abstract}

\section{Introduction}

Parity games are two-player games played on directed graphs with ranked vertices
$\ParityGameFullDef{\G}{\V}{\E}{\R}$. Player $\PX$ makes a move at vertices from
$\V_\PX$ by moving a token along one of the outgoing edges.
By making moves, players form a sequence of vertices called a play.
A play may be finite or not, but it must be
exhaustive. That is, whenever there is a possibility to make a move, then a move is made.
Player $\PE$ wins a play $\Play$ if it is either finite and ends at a vertex from $\V_\PO$ or is
infinite and the parity of the highest rank that occurs infinitely often is even,
otherwise player $\PO$ wins.

Parity games play an important role in system verification and synthesis \cite{practical_synthesis}.
However, the exact complexity of solving parity games is still an open problem.
The fastest known algorithms work in time $n^{\mathcal{O}\left(\log\left(1 + \frac{d}{\log(n)}\right)\right)}$ \cite{Lehtinen_2022, 8005092}, where $n$ is the size of the
graph and $d$ is the number of different ranks.
The problem was also tackled from the point of view of parameterized complexity.
\Obrzdalek{} showed that in graphs with bounded tree-width \cite{10.1007/978-3-540-45069-6_7} or clique-width \cite{10.5555/2392389.2392399}
parity games can be solved in polynomial time. 
The result for tree-width was later improved to $\NCTWO$ by Fearnley and Schewe \cite{10.1007/978-3-642-31585-5_20} and to $\LOGCFL$ by Ganardi \cite{10.1007/978-3-662-46678-0_25}.
The result for clique-width was improved by Ganardi to $\LOGCFL$ under an additional assumption that the $k$-expression is given \cite{10.1007/978-3-662-46678-0_25}.
It was also shown that parity games admit a polynomial-time algorithm on graphs with bounded entanglement \cite{10.1007/978-3-540-32275-7_15}, DAG-width \cite{10.1007/11672142_43}, or Kelly-width \cite{HUNTER2008206}
(for Kelly-width, decomposition is assumed to be a part of the input). For more about these classes and relations between them see \cite{GANIAN201488}.
Calude, Jain, Khoussainov, Li, and Stephan in their breakthrough work about solving parity games in quasipolynomial time \cite{10.1145/3055399.3055409}
showed that the problem is fixed-parameter tractable when parameterized by the number of ranks.

In this work, we show that solving parity games on graphs of bounded tree-depth is in logspace uniform $\ACO$
\footnote{A problem is said to be in logspace uniform $\ACO$ if there is a constant $k$, a polynomial $p$, and a deterministic Turing machine that given an input size $N$, generates in
space $\mathcal{O}(\log(N))$ the description of a circuit of size $\mathcal{O}(p(N))$ and depth at most $k$ that solves the problem instances of size $N$.
The circuit can be viewed as a directed graph with vertices of in-degree $0$ denoted as input, out-degree $0$ as output, and other vertices 
labeled with boolean operations ($\wedge$, $\vee$ or $\neg$; both $\wedge$ and $\vee$ can take an arbitrary number of arguments).
Edges show computation flow. In such a graph, vertices are usually called gates.
For a more detailed definition see \cite{arora_barak_2009}.}.
We start by introducing a parameter, which we call shallow clique-width, in \cref{sec:def_not}.
In \cref{sec:tools}, we introduce tools that are subsequently used in \cref{sec:circ_for_scw} to create an $\ACO$ algorithm for parity games on
graphs of bounded shallow clique-width.
In \cref{sec:circ_td}, we show how to efficiently reduce the problem of solving parity games on graphs of bounded tree-depth
to solving parity games on graphs of bounded shallow clique-width.
Our result improves over the ones of
\Obrzdalek{} and Ganardi when considering graphs of bounded tree-depth.
That is because logspace uniform $\ACO$ is a strict subclass of $\LOGCFL$ (parity is not in $\ACO$ \cite{ac0parity_sips,Ajtai198311FormulaeOF}),
and graphs of bounded tree-depth have bounded tree-width.
In addition to this, our intermediate result for shallow clique-width is an example of an $\ACO$ algorithm for solving parity games 
on a class of graphs that are not sparse (assuming an appropriate graph description is provided).

\subsection*{Related Work}
  The relation between uniform $\ACO$ and tree-depth was studied in \cite{10.1145/2946799,elberfeld:hal-00678176}.
  Elberfeld, Jakoby, and Tantau showed that $\MSO$ definable problems are in uniform $\ACO$ when
  considering structures whose Gaifman graphs have bounded tree-depth \cite{elberfeld:hal-00678176}.
  Elberfeld, Grohe, and Tantau showed that on graphs of bounded tree-depth $\FOL$, $\MSO$, and $\GSO$ (guarded second-order logic)
  have the same expressive power \cite{10.1145/2946799}.
  The question of whether the winning regions of parity games are definable in $\MSO$ is unsolved for the case when the number of ranks is unbounded.
  Dawar and Gr{\"a}del showed that they are not definable in $\mu$-calculus (a fragment of $\MSO$),
  but can be defined in $\GSO$ if we use a quasi-order over vertices (having a higher rank) along with a unary predicate $\mathrm{OddRank}$ \cite{10.1007/978-3-540-87531-4_26}.
  However, this along with \cite{10.1145/2946799,elberfeld:hal-00678176} does not place parity games on graphs of bounded tree-depth in uniform $\ACO$
  since the tree-depth of the Gaifman graphs of the structures considered in \cite{10.1007/978-3-540-87531-4_26} is unbounded due to the aforementioned quasi-order.

 We use a dynamic programming approach along the graph decomposition similar to the one of \Obrzdalek{} \cite{10.5555/2392389.2392399}. 
  However, dealing with $\ACO$ requires more technical care. In particular, we need to provide definitions of operations that are computable in logspace uniform $\ACO$.
 Another difference is that we don't restrict our attention to $t$-strategies but to a potentially larger set of positional strategies. We also provide the efficient reduction
 from solving parity games on graphs of bounded tree-depth to solving parity games on graphs of bounded shallow clique-width.

\section{Definitions and Notation}\label{sec:def_not}
\subsection{Basic Notation}\label{sec:bas_not}

We use $g=\FunAugment{f}{y}{c}$ to define $g$ that coincides with $f$ except that $g(y) = c$.
To define a partial function that coincides with $f$ but is undefined on $y$ we write $\FunAugment{f}{y}{\UndefinedM}$.
By $\Domain{f}$ we denote the set of elements for which $f$ is defined.
Given a function $f: A \FT B$ and finite or infinite sequence $s = a_0, a_1, \ldots$ of elements of $A$,
we denote application of $f$ 
to $s$ by $f(s) = f(a_0), f(a_1), \ldots$.

\subsection{Arenas, Games, Colorings}
An \emph{arena} is a directed graph $\GameGraphFullDef{\G}{\V}{\E}$ with vertices partitioned between players $\PE$ and $\PO$.
In a parity game $\ParityGameFullDef{\G}{\V}{\E}{\R}$, the behavior of player $\PX$ is modeled by a partial function 
$\ST_\PX : \V^* \times \V_\PX \PFT \V$
called \emph{strategy}, 
which given the sequence of visited vertices and the vertex that the token is on, gives
the next vertex to move the token to.
Every strategy must be exhaustive and correct -- whenever there is a possibility of player movement,
then the function must be defined,
and the token can only be moved along existing edges.
A strategy $\ST_\PX$ is \emph{positional} if the result depends only on the 
vertex the token is on (i.e. for a vertex $v \in \V_\PX$ and a prefix of a play $\PlayPref$ we have $\ST_\PX(\PlayPref, v) = \ST_\PX(v)$).
A play $\Play$ is consistent with $\ST_\PX$ if whenever it is the turn of player $\PX$,
then the move is made according to $\ST_\PX$.
A strategy $\ST_\PX$ is winning for player $\PX$ from a
vertex $v$ if all plays starting at $v$ and consistent with $\ST_\PX$ are winning for player $\PX$.
It is known that parity games are globally positionally determined \cite{185392_parity_pos_determ}. That is for $\ParityGameDef{\G}{\V}{\E}{\R}$
there exists a partition of $\V$ into the winning regions $\WinS{\PE}$ and $\WinS{\PO}$ such that
$\PE$ has a positional strategy $\ST_\PE$ that wins from every vertex in $\WinS{\PE}$, and
$\PO$ has a positional strategy $\ST_\PO$ that wins from every vertex in $\WinS{\PO}$.

We define $\PlayFrom{\G}{\ST_\PX}{v}$ as the set of plays in the arena $\G$ that start at the vertex 
$v$ and are consistent with the strategy $\ST_\PX$ for player $\PX$,
whereas by $\PlayFromTo{\G}{\ST_\PX}{v}{w}$
we define their prefixes that end at $w$ (i.e. $\PlayFromTo{\G}{\ST_\PX}{v}{w} = \{\Play[\ldots i] : \Play[i] = w \wedge \Play \in \PlayFrom{\G}{\ST_\PX}{v}\}$).
To distinguish a play from its prefix, we denote plays using the symbol $\Play$ and prefixes of plays using 
$\PlayPref$.

We say that $\Color: \V \FT \CPair{}$ is a player-aware coloring of vertices from
$\GameGraphAltFullDef{\G}{\V}{\E}$ if vertices from $\V_X$ are colored using colors from $\C_\PX$.
We will sometimes write a parity game along with player-aware coloring of its arena and say that 
$\ParityGameDefC{\G'}{\V'}{\E'}{\R'}{\Color'}$  is a subgame of  $\ParityGameDefC{\G}{\V}{\E}{\R}{\Color}$
if and only if $\V'_\PE \subseteq \V_\PE$, $\V'_\PO \subseteq \V_\PO$, $\E' \subseteq \E$, 
$\R'$ is $\R$ with domain restricted to $\V'$ and $\Color'$ is $\Color$ with domain restricted to $\V'$.

\subsection{Shallow Clique-Width for Arenas}
Here we define shallow clique-width for arenas, which
can be viewed as clique-width \cite{COURCELLE1993218,COURCELLE200077} with unions of arbitrary number of components (instead of only binary ones)
and restricted term height.

\begin{definition}
  A \emph{tree-model} $\TreeModelShort$ consists of a rooted tree $\T$, with all root-to-leaf paths having the same length, a set of colors $\C = \C_\PE \uplus \C_\PO$, a coloring
  $\Color$ of leaves of $\T$, and a function $\TCS$ that assigns to each internal node of $\T$ a subset of $\C\times\C$.
\end{definition}

    For a tree-model $\TreeModelShort$ and a node $l$ of $\T$, by $\TreeModelName_l$ we mean a
    tree-model $\TreeModelTupleLong{\T_l}{\C_\PE \uplus \C_\PO}{\Color'}{\TCS'}$ where $\T_l$ is the subtree of $\T$
    rooted at $l$, $\Color'$ is $\Color$ with domain restricted to leaves of $\T_l$, and $\TCS'$ is $\TCS$ with domain 
    restricted to internal nodes of $\T_l$.

\begin{definition}
    We define an arena $\GameGraphFullDef{\G}{\V}{\E}$ induced by $\TreeModelShort$ in the following way.
    \begin{itemize}
      \item If $\T$ consists of one node $v$ and $\Color(v) \in \C_\PX$, then it induces edgeless $\G$ with $\V_\PX = \{v\}$.
      \item Otherwise, let $l$ be $\T$'s root and let $l_1, \ldots l_m$ be children of $l$. 
      The arena induced by $\TreeModelName$ is created by taking union of arenas induced by 
      $\TreeModelName_{l_1}, \ldots, \TreeModelName_{l_m}$, and
       adding a directed edge from every vertex colored $c$ to every vertex colored $c'$ whenever
      $\Pair{c}{c'} \in \TCS(l)$.
    \end{itemize} 
\end{definition}

For an example of a tree-model and the induced arena see \cref{fig:ef_tm_example}.

\begin{definition}
  The \emph{shallow clique-width} of an arena $\GameGraphFullDef{\G}{\V}{\E}$ is the smallest $\Param$
  such that there exists a tree-model $\TreeModelShort$ with 
  $|\C| \leq \Param$ and $\T$ of height at most $\Param$ such that the induced arena is $\G$.
\end{definition}

One might observe that shallow clique-width is similar to the notion of shrub-depth \cite{shrub_intro,lmcs:5149}.
The main difference, however,
is that for a shallow clique-width's tree-model $\TreeModelShort$ and an internal node $l$ of $\T$ a pair $\Tuple{c, c'} \in \TCS(l)$
affects all pairs $\Tuple{v, v'}$ of leaves in the subtree of $\T$ rooted at $l$ such that $\Color(v) = c$ and $\Color(v')=c'$,
whereas in shrub-depth it is additionally restricted to only those for which the lowest common ancestor  in $\T$ is $l$.

It is worth mentioning that if a class of graphs has bounded shallow clique-width, then it 
has bounded clique-width and shrub-depth.
The former observation follows from the definition of clique-width,
the latter one from the fact that in the shallow clique-width's tree-model there are at most $2^{|\C\times\C|}$ sets $\TCS(l)$
at each level of the tree-model. So at the expense of an increase in the number of colors, we can get rid 
of these sets and describe connections using only the distance to the lowest common ancestor and colors.

\subsection{Tree-Depth}\label{ssec:td_intro}
Tree-depth was originally defined for undirected graphs \cite{treedepth_def}. 
Here we adapt this definition to arenas by simply forgetting the orientation of edges.
\begin{definition}
    A forest of rooted trees $\Forest$ is an \emph{elimination forest} of an arena $\GameGraphFullDef{\G}{\V}{\E}$
    if vertices of $\Forest$ are vertices of $\G$ and
    $\Pair{v}{w} \in \E$ implies that $v$ and $w$ are in the same tree $\Tree$ in $\Forest$
    and either $v$ is an ancestor of $w$ or $w$ is an ancestor of $v$ in $\Tree$ (see \cref{fig:ef_tm_example}).
    Tree-depth of an arena $\G$ is the smallest height 
    of an elimination forest of $\G$.
\end{definition}
\begin{restatable}{lemma}{LemmaTreeDepthShortPaths}\label{lemma:treedepth_short_paths}\cite[Chapter 6.2]{sparsity}
    If an arena $\G$ has an elimination forest of height $k$, then any simple path in $\G$
    has length at most $2^{k+1} - 2$.
\end{restatable}
The lemma above can be proven by induction on the height of the elimination forest.

One can think of graphs of bounded tree-depth as graphs of bounded tree-width \cite[Chapter 12]{10.5555/3134208} with an 
additional constraint on the height of the tree-decomposition. More precisely, if
a graph $\G$ has a tree-decomposition of height $h$ where each bag contains at most $k$ vertices,
then it has tree-depth at most $k(h+1)$ \cite[Chapter 6.4]{sparsity} (to obtain the elimination forest from the tree-decomposition, we leave each vertex of $\G$ in the bag that is closest to 
the root of the tree-decomposition and then replace each bag in the tree-decomposition with a path of its elements).
What is more, if a graph has tree-depth $k$, then it admits a tree-decomposition of height bounded by $k + 1$ where each bag contains at most $k + 2$ vertices
(to obtain the tree-decomposition, we take the elimination forest $\Forest$, 
make it a tree $\Tree$ by choosing a root $r$ of an arbitrary tree in $\Forest$ and connecting other roots to $r$, 
and define bags for vertices of $\Tree$ as sets of vertices on vertex-$r$ paths).

  Shallow clique-width is more general than tree-depth.
  First, for an arena $\G$ of tree-depth $k$ one can use $2(k + 1)3^{k + 1}$ colors to create a tree-model of height $k + 1$
  that induces $\G$. This is because we can take some elimination forest $\Forest$ of $\G$ of height $k$ and color each vertex according to its player-membership in $\G$, and depth and connections to ancestors in $\Forest$
  (for each vertex $v$ and its ancestor $w$ we have that either there is a directed edge from $v$ to $w$, to $v$ from $w$, or $v$ and $w$ are not connected in $\G$).
  On the other hand, arenas where each vertex has a directed edge to every vertex have small shallow clique-width, but large tree-depth.

  The major advantage of tree-depth over shallow clique-width is that given an arena $\G$ of tree-depth at most $k$, we know how to efficiently find 
  an elimination forest of $\G$ of height at most $k$ (details in \cref{sec:EX}).
  It is also worth noting that despite the simplicity of tree-depth, some problems are $\WOHard$ when parameterized by it \cite[Lemma 2, Proposition 5]{DBLP:journals/corr/GajarskyLO13}.
  
  As mentioned before, in this work, we take a detour by first proving that solving parity games played on arenas of bounded
  shallow clique-width is in logspace uniform $\ACO$ (assuming an appropriate graph description is provided) and then providing a reduction.
  We decided to do so as thinking about vertices of the same color as the ones that will be treated similarly 
  turned out to be conceptually simpler from our perspective. What is more, such an approach gave us a 
  logspace uniform family of $\ACO$-circuits for a class of graphs that are not sparse.

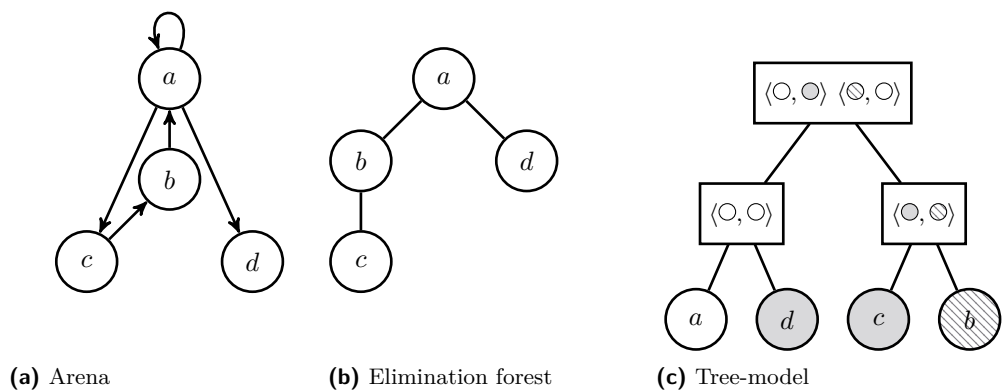
\begin{figure}
  \centering
  \begin{subfigure}{.3\textwidth}
      \centering
      \begin{tikzpicture}[>=stealth', line width=1pt, minimum size=0.8cm]
        \node at (0, 2) (z) {};
        \node at (0, -2) (z2) {};
        \node[draw, circle] at (0, 1.2) (a) {$a$};
      \node[draw, circle] [below = 0.5 of a] (b) {$b$};
      \node[draw, circle] [below left = 0.5 and 0.5 of b] (c) {$c$};
      \node[draw, circle] [below right = 0.5 and 0.5 of b] (d) {$d$};
  
      \draw[->] (a) -- (d);
      \draw[->] (a) -- (c);
      \draw[<-] (a) -- (b);
      \draw[<-] (b) -- (c);

      \draw[->] (a) to [out=70,in=110,looseness=6]  (a);

      \end{tikzpicture}
      \subcaption{Arena}
    \end{subfigure}%
    \begin{subfigure}{.3\textwidth}
      \begin{tikzpicture}[>=stealth', line width=1pt, minimum size=0.8cm]
        \node at (0, 2) (z) {};
        \node at (0, -2) (z2) {};
        \node[draw, circle] at (0, 1.2) (a) {$a$};
        \node[draw, circle] [below left = 0.5 and 0.5 of a] (b) {$b$};
        \node[draw, circle] [below = 0.5 of b] (c) {$c$};
        \node[draw, circle] [below right = 0.5 and 0.5 of a] (d) {$d$};
  
        \draw (b) -- (a);
        \draw (c) -- (b);
        \draw (d) -- (a);

      \end{tikzpicture}
      \subcaption{Elimination forest}
    \end{subfigure}
    \begin{subfigure}{.33\textwidth}
      \centering
      \begin{tikzpicture}[>=stealth', line width=1pt, minimum size=0.8cm]
        \node at (0, 4) (z) {};
        \node at (0, 0) (z2) {};
        \node[draw, circle] at (0, 0) (a) {$a$};
          \node[draw, circle, fill=lipicsLightGray] at (1.2, 0) (b) {$d$};
          \node[draw, circle, fill=lipicsLightGray] at (2.4, 0) (c) {$c$};
          \node[draw, circle, pattern=north west lines, pattern color=lipicsLineGray] at (3.6, 0) (d) {$b$};
  
          \node[draw] at (0.6, 1.4) (n1) {\makecell[c]{
          $\Pair{\WCircle}{\WCircle}$}};

          \node[draw] at (3, 1.4) (n2) {\makecell[c]{
          $\Pair{\GCircle}{\SCircle}$}};

          \node[draw] at (1.8, 3.0) (n3) {\makecell[c]{
          $\Pair{\WCircle}{\GCircle}$
          $\Pair{\SCircle}{\WCircle}$}};

          \draw (a) -- (n1);
          \draw (b) -- (n1);

          \draw (c) -- (n2);
          \draw (d) -- (n2);
  
          \draw (n1) -- (n3);
          \draw (n2) -- (n3);

      \end{tikzpicture}
      \subcaption{Tree-model}
    \end{subfigure}
  \caption{Arena (a), an elimination forest for the arena (b) and a tree-model that induces the arena (c) (all vertices belong to one player).}\label{fig:ef_tm_example}
\end{figure}

\section{Tools}\label{sec:tools}
Now we proceed with tools for constructing a logspace uniform family of circuits for 
parity games played  on arenas with
bounded shallow clique-width.
The approach presented here is similar to the one used in \cite{10.5555/2392389.2392399} for creating a polynomial-time algorithm for parity games
played on arenas of bounded clique-width.
Here we also operate on some succinct information characterizing strategies that we call enforcements.
Enforcements are an adaptation of \Obrzdalek's borders \cite{10.1007/978-3-540-45069-6_7} from tree-width to  shallow clique-width.
The main idea behind borders is to store for a separating set $S$, a strategy $\ST_\PE$, and a vertex $v$ information about the worst arrivals at vertices
of $S$
that opponent can achieve when starting from $v$ and playing against $\ST_\PE$,
whereas in enforcements we are interested in knowing what the worst arrivals at vertices of a specific color can be.
Our adaptation of borders is analogous to the one considered in \cite{10.5555/2392389.2392399}.

We 
additionally extend and modify this approach
 to obtain a logspace uniform family of constant-depth circuits.

\subsection{Parity Games With Stop}

\begin{definition}
  We define a \emph{stop parity game} as a parity game where player $\PE$ has an additional
  move called $\STOP$ at his positions.
  Executing this move ends the game in a draw. 
\end{definition}

\begin{observation}\label{observation:safe_default_pairty_win}
  In an arena $\GameGraphFullDef{\G}{\V}{\E}$ with vertices ranked by $\R : \V \FT \Nat$,
  winning regions for player $\PE$ in both stop and normal parity games
  are equal. What is more if a strategy $\ST_\PE$ is winning in either of the games,
  then it is winning in the other one, as $\ST_\PE$ cannot be forced to use $\STOP$.
\end{observation}

\begin{definition}
  We say that a strategy $\ST_\PE$ is \emph{safe from} $v$ in a stop parity game $\G$,
  if all plays that start at $v$ and are consistent with $\ST_\PE$ are not losing for player $\PE$.
\end{definition}
Stop parity games will turn out to be useful for solving parity games on arenas that can be extended.

\subsection{Enforcements}
\begin{definition}
  Let $\Values(s)$ denote the set of vertices that appear in the sequence $s$.
  For a game $\ParityGameDef{\G}{\V}{\E}{\R}$, we define a function
  $\PlayMaxRank{\G}(\PlayPref) = \max(\Values(\R(\PlayPref)))$.
  That is for a prefix of a play, return the maximum rank of its vertex.
\end{definition}

\begin{definition}
  We define a total order $\RCMP$ on natural numbers as $p \RCMP q \IFF (-1)^pp < (-1)^qq$ and 
  $p \RCMPEQ q \IFF (p \RCMP q  \vee p = q)$. That is $\ldots \RCMP 5 \RCMP 3 \RCMP 1 \RCMP 0 \RCMP 2 \RCMP 4 \RCMP 6 \RCMP \ldots$.
\end{definition}

\begin{definition}
  For a stop parity game $\ParityGameDef{\G}{\V}{\E}{\R}$, a strategy $\ST_\PE$ and $v \in \V$
  we define a partial function 
  \begin{align*}
    \EnforcementV{\G}{\ST_\PE}{v}(w)=\begin{cases}
      \MinB{\RCMPEQ}{\{\PlayMaxRank{\G}(\PlayPref) : \PlayPref \in \PlayFromTo{\G}{\ST_\PE}{v}{w}  \wedge \ST_\PE(\PlayPref) = \STOP} \} &\MIF w \in \V_\PE.\\
      \MinB{\RCMPEQ}{\{\PlayMaxRank{\G}(\PlayPref) : \PlayPref \in \PlayFromTo{\G}{\ST_\PE}{v}{w} \}} &\MIF  w \in \V_\PO
    \end{cases}
  \end{align*}
  Where $\MinB{\RCMPEQ}(A)$ is a partial function that for a finite subset of natural numbers $A$
  gives the smallest element of $A$ according to $\RCMPEQ$ and is undefined for $A = \emptyset$.
\end{definition}
Intuitively for $w \in \V_\PO$ the operation above returns the worst possible arrival at $w$ that player $\PO$ can force when playing 
against the strategy $\ST_\PE$ and for $w \in \V_\PE$  the worst arrival at $w$ such that  $\ST_\PE$
will decide to stop the game.

\begin{definition}
  We naturally adapt the definition above to games $\ParityGameDef{\G}{\V}{\E}{\R}$ with player-aware coloring $\Color: \V \FT \CPair{}$:
  \begin{align*}
    \EnforcementC{\G}{\TreeModelName}{\ST_\PE}{v}(c) = \MinB{\RCMPEQ}{\{\EnforcementV{\G}{\ST_\PE}{v}(w) : w \in \Domain{\EnforcementV{\G}{\ST_\PE}{v}} \wedge \Color(w) = c\}}
  \end{align*}
\end{definition}

  An \emph{enforcement} is any partial function from $\C$ to $\{0,\ldots,\max(\R(\V))\}$.
  From now on by 
  \emph{enforcement from $v$}
   of a strategy $\ST_\PE$ in $\G$
   we will mean
  $\EnforcementC{\G}{\TreeModel}{\ST_\PE}{v}$. Note that this is always an enforcement
  and that $\ST_\PE$ can have different enforcements from different vertices (see \cref{fig:enf_example}).
  
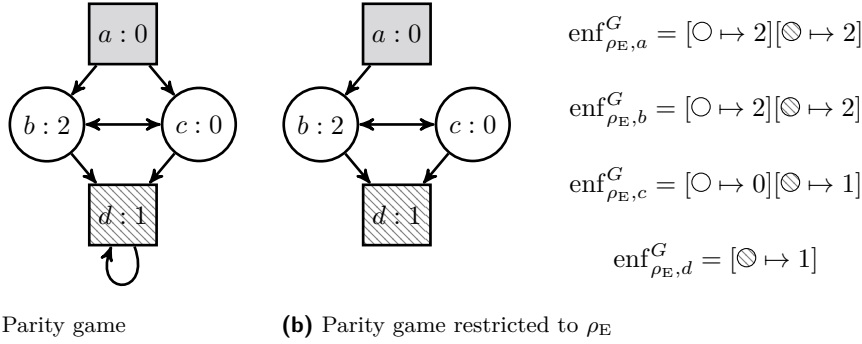
\begin{figure}
    \centering
    \begin{subfigure}{.3\textwidth}
        \centering
        \begin{tikzpicture}[>=stealth', line width=1pt, minimum size=0.8cm]
          \node[draw, fill=lipicsLightGray] at (0, 0) (s) {$a:0$};
          \node[draw, circle] at (-1, -1.2) (a) {$b:2$};
          \node[draw, circle] at (1, -1.2) (b) {$c:0$};

          \node[draw, pattern=north west lines, pattern color=lipicsLineGray] at (0, -2.4) (t) {$d:1$};

          \draw[->] (s) -- (a);
          \draw[->] (s) -- (b);
          \draw[->] (a) -- (t);
          \draw[->] (b) -- (t);
          \draw[<->] (b) -- (a);

          \draw[->] (t) to [out=290,in=250,looseness=6]  (t);

        \end{tikzpicture}
        \subcaption{Parity game}
      \end{subfigure}%
      \begin{subfigure}{.7\textwidth}
        \begin{tikzpicture}[>=stealth', line width=1pt, minimum size=0.8cm]
          \node[draw, fill=lipicsLightGray] at (0, 0) (s) {$a:0$};
          \node[draw, circle] at (-1, -1.2) (a) {$b:2$};
          \node[draw, circle] at (1, -1.2) (b) {$c:0$};

          \node[draw, pattern=north west lines, pattern color=lipicsLineGray] at (0, -2.4) (t) {$d:1$};

          \draw[->] (s) -- (a);
          \draw[->] (a) -- (t);
          \draw[->] (b) -- (t);
          \draw[<->] (b) -- (a);

          \phantom{\draw[->] (t) to [out=290,in=250,looseness=6]  (t);}
  
          \node at (4.2, 0) {$\EnforcementC{\G}{\TreeModelName}{\ST_\PE}{a} = \FunAugment{\FunAugment{}{\WCircle}{2}}{\SCircle}{2}$};
          \node at (4.2, -1) {$\EnforcementC{\G}{\TreeModelName}{\ST_\PE}{b} = \FunAugment{\FunAugment{}{\WCircle}{2}}{\SCircle}{2}$};
          \node at (4.2, -2) {$\EnforcementC{\G}{\TreeModelName}{\ST_\PE}{c} = \FunAugment{\FunAugment{}{\WCircle}{0}}{\SCircle}{1}$};
          \node at (4.2, -3) {$\EnforcementC{\G}{\TreeModelName}{\ST_\PE}{d} = \FunAugment{}{\SCircle}{1}$};
        \end{tikzpicture}
        \subcaption{Parity game restricted to $\ST_\PE$}
      \end{subfigure}
    \caption{A parity game $\G$ (a), $\G$ with the moves of player $\PE$ restricted to the positional strategy $\ST_\PE = \FunAugment{\FunAugment{}{\Pair{\PlayPref}{a}}{b}}{\Pair{\PlayPref}{d}}{\STOP}$ 
            along with the enforcements of $\ST_\PE$ from various vertices (b).
            Square vertices belong to player $\PE$, round ones to $\PO$. 
            The letters denote vertices, whereas the numbers inside vertices denote their ranks.
            Background shading denotes vertex color (vertices $b$ and $c$ have the same color, whereas colors of $a$, $b$ and $d$ are pairwise distinct).
            }\label{fig:enf_example}
\end{figure}

\begin{observation}\label{observation:enf_no_stop_c}
  Enforcement $\EnforcementC{\G}{\TreeModelName}{\ST_\PE}{v}$ is defined for some color of player $\PE$ 
  (i.e. color $c$ such that  $c \in \C_\PE$)
  if and only if there is a finite play $\Play$ consistent with $\ST_\PE$ that starts at $v$
  such that $\ST_\PE(\Play) = \STOP$.
\end{observation}

\begin{definition}
  We define a partial order $\ECMPEQ$ on enforcements as follows:\\
  $P \ECMPEQ Q \IFF \Domain{P} \subseteq \Domain{Q} \wedge \FA{a \in \Domain{P}}{Q(a) \RCMPEQ P(a)}$\\
  $P \ECMP Q \IFF  P \ECMPEQ Q \wedge P \not= Q$.
\end{definition}
Intuitively $P \ECMP Q$ means that $Q$ describes worse arrival scenarios from player $\PE$ perspective than $P$.

\subsubsection*{Operations on Enforcements}

Clearly, modifications of a strategy $\rho_\PE$ 
can result in modifications of enforcements of this strategy.
What is more, some of these modifications can be done directly on the enforcements without knowing
the original strategy.
In the further parts of this work, we will be interested in 
modifications that adapt strategies of player $\PE$ from an arena $\G$ 
to an arena $\G'$, that was made from $\G$ by addition of directed edges
between vertices of specified colors.
It will usually suffice to remember enforcements of a 
strategy instead of the strategy itself and to perform respective operations
on enforcements. Below we introduce some useful operations on enforcements.

\bigskip
The first operation
corresponds to a modification of a strategy $\ST_\PE$ to the strategy $\ST'_\PE$
that behaves like $\ST_\PE$, but each time $\ST_\PE$ decides to $\STOP$ at some vertex of a specified color
$c$ $\ST'_\PE$ uses the option to move to the vertex from which the game was started. Moreover, we want to ensure that $\ST'_\PE$ 
will never close a cycle with the highest rank being odd by doing so.
\begin{definition}
  For $\ParityGameFullDefC{\G}{\V}{\E}{\R}{\Color}$,
  enforcement $P$ and color $c$ we define
  \begin{align*}
    \EnfLoop{P}{c} =\begin{cases}
      P  & \MIF c \in \C_\PE \wedge  c \not\in \Domain{P}\\
      \FunAugment{P}{c}{\UndefinedM} &\MIF c \in \C_\PE \wedge P(c) \text{ is even }
    \end{cases}
  \end{align*}
\end{definition}
Please note that $\EnfLoop{P}{c}$ is undefined if $c \in \C_\PO$ or
$P(c)$ is odd. 
  Let $P$ be the enforcement of $\ST_\PE$ from some vertex $v$.
Let $c$ be any element of $\C_\PE$ and assume that from 
each vertex of color $c$ there is an edge to $v$.
Now, the first case in the definition of $\EnfLoop{P}{c}$ corresponds
to the situation when $\ST_\PE$ will never stop at a vertex of color $c$
when starting the game from $v$. Whereas the second one to the situation 
when $\ST_\PE$ can be forced to stop at a vertex of color $c$, but the prefix 
of the play that will lead here will always have even max rank.

\begin{definition}
  For $\ParityGameFullDefC{\G}{\V}{\E}{\R}{\Color}$, $r \in \Nat$,
  color $c$ and enforcements $P$, $Q$, we define two partial functions
  \begin{align*}
    \EnfLift{Q}{r}(c) &= \max(Q(c), r) \;\MIF c \in \Domain{Q}\\
    \EnfUnion{P}{Q}(c) &= \Min{\RCMPEQ}(P(c), Q(c))
  \end{align*}
  In this writing, we allow one of $\Min{\RCMPEQ}(P(c), Q(c))$ arguments to be \UndefinedT{} 
  and define $\Min{\RCMPEQ}(P(c), Q(c))$ to be the defined one (or \UndefinedT{} when both are \UndefinedT{}).
  The $\max(Q(c), r)$ in the definition of $\EnfLift{Q}{r}(c)$ is taken using the standard order on $\Nat$.
\end{definition}

Another operation is for composing strategies
in a way that corresponds to playing one strategy and switching to another when either reaching 
vertices of specified color $c \in \C_\PO$ or stopping at vertices of specified color $c \in \C_\PE$.
\begin{definition}
  For $\ParityGameFullDefC{\G}{\V}{\E}{\R}{\Color}$,
  color $c$ and enforcements $P$, $Q$, we define
  \begin{align*}
    \EnfMerge{\Color}{P}{c}{Q} = 
    \begin{cases}
      \EnfUnion{P}{\EnfLift{Q}{P(c)}} & \MIF c \in \Domain{P} \wedge c \in \C_\PO \\
        \EnfUnion{\FunAugment{P}{c}{\UndefinedM}}{\EnfLift{Q}{P(c)}} & \MIF c \in \Domain{P}  \wedge c \in \C_\PE\\
        P & \MOTHERWISE
    \end{cases}
  \end{align*}
\end{definition}
Intuitively, taking $\EnfLift{Q}{P(c)}$ in the formula above corresponds to extending plays that were used 
to calculate $P$ with the ones being used to calculate $Q$. More detailed explanation 
is the following.

\begin{lemma}\label{lemma:properties_of_merge}
  Properties of $\EnfMerge{}{P}{m}{Q}$ and $\EnfLoop{P}{m}$:
  \begin{enumerate}
    \item Let $P, P', Q, Q'$ be enforcements such that 
    $P' \ECMPEQ P$ and $Q' \ECMPEQ Q$, then
    \[\EnfMerge{\Color}{P'}{c}{Q'} \ECMPEQ \EnfMerge{\Color}{P}{c}{Q}\]
    \item 
    Let $c$ be a color of player $\PO$.
    Let $\ST_\PE$, $\ST'_\PE$, be
    such that $\ST_\PE$ is safe from $v$, whereas $\ST'_\PE$ is safe from every vertex  
    of color $c$ in $\G$. Let $\ST^{\prime\prime}_\PE$ be like
    $\ST_\PE$, but when $\ST_\PE$ reaches a vertex of color $c$,
    $\ST^{\prime\prime}_\PE$ forgets about the past and switches permanently to $\ST'_\PE$ behavior.
    Then $\ST^{\prime\prime}_\PE$ is safe from $v$ in $\G$ and
    \[\EnforcementC{\G}{\TreeModelName}{\ST^{\prime\prime}_\PE}{v} \ECMPEQ \EnfMerge{\Color}{\EnforcementC{\G}{\TreeModelName}{\ST_\PE}{v}}{c}{\EnfUnion{\EnforcementC{\G}{\TreeModelName}{\ST'_\PE}{w_1},\ldots}{\EnforcementC{\G}{\TreeModelName}{\ST'_\PE}{w_m}}}\]
    where $w_1, \ldots, w_m$ are all vertices of color $c$ and $\EnfUnion{A, B}{C} = \EnfUnion{\EnfUnion{A}{B}}{C}$.
    \item Let $\ST_\PE$, $\ST'_\PE$, be strategies that are safe from $v$ and $w$ in $\G$ respectively.
    Let $c$ be a color of player $\PE$ such that from every vertex of color $c$ there is an edge to $w$,
    and let $\ST^{\prime\prime}_\PE$ be like
    $\ST_\PE$, but when $\ST_\PE$ decides to stop at some vertex of color $c$,
    $\ST^{\prime\prime}_\PE$ moves to $w$ instead, forgets about the past and switches permanently to $\ST'_\PE$ behavior.
    Then $\ST^{\prime\prime}_\PE$ is safe from $v$ in $\G$ and
    \[\EnforcementC{\G}{\TreeModelName}{\ST^{\prime\prime}_\PE}{v} = \EnfMerge{\Color}{\EnforcementC{\G}{\TreeModelName}{\ST_\PE}{v}}{c}{\EnforcementC{\G}{\TreeModelName}{\ST'_\PE}{w}}\]
    \item Let $\ST_\PE$ be a strategy
    that is safe from $v$ in $\G$.
    Let $c$ be a color of player $\PE$
    such that from every vertex of color $c$ there is an edge to $v$,
    and let $\ST'_\PE$ be like 
    $\ST_\PE$, but each time  $\ST_\PE$ decides to stop at some vertex of color $c$,
    $\ST'_\PE$ moves to $v$, forgets about the past and resumes as $\ST'_\PE$.
    Then if 
    $\EnfLoop{\EnforcementC{\G}{\TreeModelName}{\ST_\PE}{v}}{c}$ is defined,
    we have that $\ST'_\PE$ is safe from $v$ in $\G$ and
    \[\EnforcementC{\G}{\TreeModelName}{\ST'_\PE}{v} = \EnfLoop{\EnforcementC{\G}{\TreeModelName}{\ST_\PE}{v}}{c}\]
  \end{enumerate}
\end{lemma}
The lemma above can be proven by simple but tedious case analysis. The key step in the proof is to show that 
$p \RCMPEQ p' \wedge q \RCMPEQ q'$ implies $\max(p, q) \RCMPEQ \max(p', q')$, what can be done by considering cases on the parity of $p'$ and $q'$.

\subsubsection*{Enforcements and Strategies}
For a set of enforcements $\EnfSetAny{\G}{v}$ and a vertex $v$ of a game $\G$ we consider the following properties.
\begin{enumerate}
  \item (\emph{soundness}) If $P \in \EnfSetAny{\G}{v}$, then there exists a (possibly non-positional) strategy $\ST_\PE$
  that is safe from $v$ in $\G$ such that $\EnforcementC{\G}{\TreeModelName}{\ST_\PE}{v} \ECMPEQ P$.
  \item (\emph{completeness}) For each positional strategy $\ST_\PE$ that is safe from $v$ in $\G$
  we have that $\EnforcementC{\G}{\TreeModelName}{\ST_\PE}{v} \in \EnfSetAny{\G}{v}$.
  \item (\emph{up-closure}) If $P \in \EnfSetAny{\G}{v}$ and $P \ECMPEQ Q$ then $Q \in \EnfSetAny{\G}{v}$.
\end{enumerate}
We will say that $\EnfSetAny{\G}{v}$ is \emph{valid} for $v$ if it is  sound, complete and up-closed for $v$.

\medskip

From now on we will assume that for a stop parity game $\ParityGameDef{\G}{\V}{\E}{\R}$
we always have $\R(V) \subseteq \{0, \ldots, 2|\R(V)| - 1\}$,  and denote $2|\R(V)|$ by $d$,
as otherwise we can sort elements of $\R(V)$ and put them in the desired range.

\begin{lemma}\label{theorem:safety_parity_enf_win}
  Given a set $\EnfSetAny{\G}{v}$ of enforcements in $\G$ that is valid for $v$, 
  we have that
   player $\PE$ wins the parity game on $\G$'s arena from $v$ if and only if 
   $\EnfSetAny{\G}{v}$ contains an enforcement that is \UndefinedT{} for every color
    of player $\PE$.
   
 \end{lemma}

 \begin{proof}
  If player $\PE$ wins from $v$ in the parity game played on $\G$'s arena, then
  there exists a positional strategy $\ST_\PE$ that wins the parity game starting at $v$, and 
  the enforcement generated by $\ST_\PE$  from $v$ in $\G$ is \UndefinedT{} for
  every color of player $\PE$ (see \cref{observation:enf_no_stop_c}). What is more, 
  since $\EnfSetAny{\G}{v}$ is valid for $v$, 
  from \cref{observation:safe_default_pairty_win} we have that $\EnforcementC{\G}{\TreeModelName}{\ST_\PE}{v} \in \EnfSetAny{\G}{v}$.
  On the other hand, if $\EnfSetAny{\G}{v}$ contains an enforcement $P$
 that is undefined for
 every color of player $\PE$, 
 then since $\EnfSetAny{\G}{v}$ is valid for $v$, there exists a strategy $\ST_\PE$ that is safe from $v$
 such that $\EnforcementC{\G}{\TreeModelName}{\ST_\PE}{v} \ECMPEQ P$, what by \cref{observation:enf_no_stop_c} means that 
 $\ST_\PE$ never uses  $\STOP$.
 \end{proof}

 In the following parts, we will aim to calculate for each vertex of a game a valid set of enforcements.

\section{Circuit Family for Shallow Clique-Width}\label{sec:circ_for_scw}
Now we will show how for $n, \NR > 1$, and $\Param > 0$ create in space $\AlgComplexitySpace$ (and therefore time $\ACAlgComplexity$) 
the circuit $\Circuit_\Tuple{n, \NR, \Param}$
of depth $\ACAlgDepth$, that given the specific representation of a parity game 
of shallow clique-width at most $k$, computes the partition into winning regions.

\subsection{Overview}\label{sec:circ_overview_brief}

Circuit $\Circuit_\Tuple{n, \NR, \Param}$ will be a combination of auxiliary circuits called gadgets.
The circuit will input parity game's $\G$ arena as a tree-model $\TreeModelShort$ (details about the input are presented in \cref{sec:ac0_shrub_circ_in}),
and it will process $\TreeModelName{}$ bottom-up.

The first layers of the circuit will calculate valid sets of enforcements for one-vertex arenas
consisting of leaves of $\T$. To do so gadgets defined in \cref{sss:init_gadget} will be used.
The following layers of the circuit will correspond to computing valid sets of enforcements for larger and larger arenas 
induced by larger and larger subtrees of the given tree-model. 
For the last layers, we will use \cref{theorem:safety_parity_enf_win} to get the partition into winning regions.

To be more precise for an internal node $l$ of $\T$ 
the circuit will consider a subgame $\G_l$ played on the arena induced by $\TreeModelName_l$.
For each vertex of such $\G_l$ the circuit will calculate a valid set of enforcements
using valid sets of enforcements calculated for subgames played on arenas induced by 
children of $l$.
This calculation will be divided into several steps.
First, a disjoint union of arenas induced by children of $l$ will be taken,
and then connections described by $\TCS(l)$ will be added one by one.
For a pair $\Tuple{s, t} \in \TCS(l)$ cases on whether $s$ is a color of player $\PE$ or $\PO$ will be considered.
Let $\G'$ be the subgame of $\G_l$ just before processing of $\Tuple{s, t}$ and $\G''$ be $\G'$ after adding directed edges from all vertices of color $s$ to all of color $t$ in $\G'$.

For the case when $s$ is a color of player $\PE$, we will observe that a positional strategy $\ST_\PE$
that is safe from a vertex $v$ in $G''$ can be modified to a strategy $\ST'_\PE$ 
such that $\ST'_\PE$  is also safe from $v$,
 $\EnforcementC{\G''}{\TreeModelName}{\ST'_\PE}{v} \ECMPEQ \EnforcementC{\G''}{\TreeModelName}{\ST_\PE}{v}$
 and $\ST'_\PE$  behaves similarly to $\ST_\PE$
but when $\ST_\PE$ decides to move along a newly added edge,
then $\ST'_\PE$ chooses to  move to some fixed vertex $w$.
Given such $\ST'_\PE$ we will be able to factor it into two strategies that 
are safe in $\G'$ from $v$ and $w$ respectively.

For the case when $s$ is a color of player $\PO$, we will observe that a positional strategy 
of player $\PE$ that is safe from $v$ in $G''$
can be factored into a positional strategy that is safe from $v$ in $\G'$ and a positional strategy 
that is safe from all vertices of color $t$ in $G'$.

These observations will allow us to calculate valid sets of enforcements for vertices of $\G''$ using valid sets of enforcements
calculated for $\G'$. These updates will be performed using gadgets defined in \cref{sec:update_gadgets}.
The whole procedure of 
adding connections between vertices of specified colors will be handled by the gadget defined in \cref{sss:proc_gadget}.
The details about connections between gadgets and therefore the construction of the whole circuit are presented 
in \cref{ss:circ_constr}.

It is worth noting that the approach presented in this section fails for parity games of bounded shrub-depth.
In short, that is because it makes use of the fact that in shallow clique-width,
we specify connections using colors, whereas in shrub-depth, we 
additionally condition connections on the lowest common ancestor in the tree-model.
To be more precise, when constructing an arena bottom-up from shallow clique-width's tree-model 
$\TreeModelShort$
during the processing of an internal vertex $l$ of $\T$
for $\Pair{\CFrom}{\CTo} \in \TCS(l)$, we have that there is an edge from each leaf of $\T_l$ of
color $\CFrom$ to each leaf of $\T_l$ of color $\CTo$. In shrub-depth, this is additionally conditioned on 
whether $l$ is the lowest common ancestor of the vertices to be connected. This means that in shrub-depth a vertex is 
 distinguished not only by its color but also by the information which subtree of $\T_l$ has the vertex as a leaf.
 Simple attempts to mitigate that either increase the depth of the circuit to be logarithmic or
 the number of possible enforcements to be exponential in the number of vertices.

\subsection{Circuit Input}\label{sec:ac0_shrub_circ_in}
The circuit $\Circuit_{\Tuple{n, \NR, \Param}}$ inputs an $n$-vertex game $\ParityGameFullDefC{\G}{\V}{\E}{\R}{\Color}$ as  
$\R$ and a tree-model $\TreeModelShort$  such that $\R(V) \subseteq \{0, \ldots d - 1\}$, $\NR \leq 2n$,
$|\C| = \Param$, the height of $\T$ is $\Param$ and $\TreeModelName$ induces the arena of $\G$. The whole input is 
encoded as follows. 
\begin{itemize}
    \item For each vertex $v \in \V$: 
    $\Param$ bits that one-hot encode its color (denoted as $\ColorB{}{v}$, one bit for one color, only the bit denoting the true color of the vertex is going to be $\True$), 
    and $\NR$ bits that one-hot encode $\R(v)$ (denoted as $\VRankB{v}$).
    \item $\Param$ bits for color-to-player mapping: $i$'th bit is set to $\True$ if $i$'th color
    belongs to player $\PE$ (denoted as $\CIsEColor{}$).
    \item Description of  $\T$ 
    viewed as a subgraph of the graph consisting of $\Param + 1$ layers, 
    where each layer consists of $n$ vertices, and 
    each of them is connected to all vertices from the layer below (see \cref{fig:node_embedder}). 
    So, for each $l \in \{1, \ldots, \Param\}$ and $i \in \{1, \ldots, n\}$:
    \begin{itemize}
        \item $n$ bits describing the connections to the layer below.
        (denoted as $\TCCB{l}{i}$, for $l=1$ it describes connections to layer $0$, which is formed of vertices from $\V$)
        \item $\Param^2$ bits describing the result of $\TCS$ for a given internal node -- one bit for each pair of colors (denoted as $\TCSB{l}{i}$).
    \end{itemize}
\end{itemize}

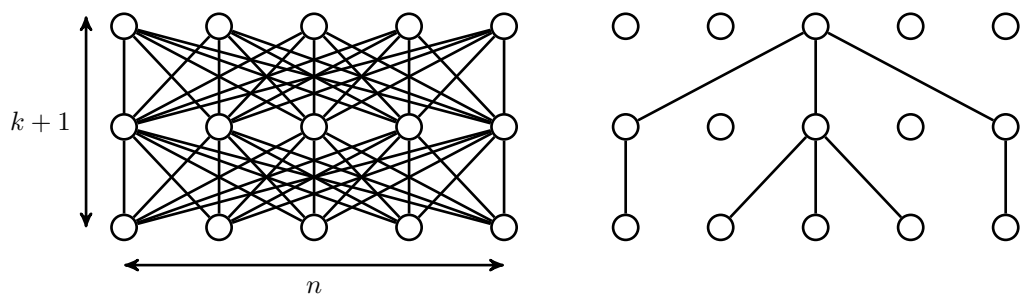
\begin{figure}
  \centering

  \begin{subfigure}{.5\textwidth}
  \begin{tikzpicture}[>=stealth', line width=1pt]
      \foreach \y in {0, ..., 2}
      \foreach \x in {0, ..., 4}
      {
          \pgfmathsetmacro{\posy}{\y*4/3}
          \pgfmathsetmacro{\posx}{\x*5/4}

          \node[draw, circle] at (\posx, \posy) (\x\y) {};   
      }
      \foreach \y in {1, ..., 2}
      \foreach \x in {0, ..., 4}
      {
          \pgfmathtruncatemacro{\ym}{\y - 1}

          \foreach \z in {0, ..., 4}
          {
              \draw (\x\y) -- (\z\ym);
          }

      }

      \draw[<->] (-0.5, 0) -- node [xshift=-6mm] {$\Param + 1$} (-0.5, 2.8);
      \draw[<->] (0, -0.5) -- node [yshift=-3mm] {$n$} (15/3, -0.5);

  \end{tikzpicture}
  \end{subfigure}%
  \begin{subfigure}{.5\textwidth}
  \begin{tikzpicture}[>=stealth', line width=1pt]
      \foreach \y in {0, ..., 2}
      \foreach \x in {0, ..., 4}
      {
          \pgfmathsetmacro{\posy}{\y*4/3}
          \pgfmathsetmacro{\posx}{\x*5/4}

          \node[draw, circle] at (\posx, \posy) (\x\y) {};   
      }

      \draw (22) -- (01);
      \draw (22) -- (41);
      \draw (22) -- (21);

      \draw (01) -- (00);
      \draw (41) -- (40);

      \draw (21) -- (20);
      \draw (21) -- (10);
      \draw (21) -- (30);

      \phantom{
          \draw[<->] (-0.1, 0) -- node [xshift=-6mm] {$\Param + 1$} (-0.1, 2.8);
          \draw[<->] (0, -0.5) -- node [yshift=-3mm] {$n$} (16/3, -0.5);
      }

  \end{tikzpicture}
\end{subfigure}

  \caption{Tree embedding graph (left) and embedding of some tree (right)} \label{fig:node_embedder}.
\end{figure}

We will call a distinguished collection of bits (which can be input bits and gate results) a bit pack.
To denote a specific bit of a bit pack we will use square brackets notation. 
For example $\ColorB{}{v}[c]$ denotes the bit that is true if and only if 
vertex $v$ has color $c$.

For validating input data, one can easily compute in space $\AlgComplexitySpaceL$ parts of the circuit that have constant-depth and validate one-hot encodings.
To validate $\T$, one can compute in space $\AlgComplexitySpaceL$ constant depth parts for validating that:
each node has at most one parent, 
each node from layer $0$ has exactly one parent, 
each node that has at least one child and is not from the last layer has a parent,
each node that has a parent and is not from the layer $0$ has at least one child,
there is exactly one node with children in the last layer.

\subsection{Gadgets}
Here we define gadgets (auxiliary circuits) that will be used to construct the whole circuit.

\subsubsection{Initialization}\label{sss:init_gadget}
Here we define a constant-depth gadget that for a vertex $v$ of $G$ calculates a valid set of enforcements for $v$ in
an edgeless subgame of $G$ that contains only one vertex $v$. We define the gadget so that it outputs empty sets of enforcements for other
vertices of $G$ and marks the vertex $v$ in its output.
\begin{definition}{$\CInitialize{v}$}
  \begin{align*}
    \VertexUsedBit{\OUT}[w] &= (w=v)\\
      \EnfB{\OUT}{w}[P] &= (w=v) \wedge \bigvee_{\Pair{c}{r}: P \text{ is in up-closure of }\{\FunAugment{}{c}{r}\}}{\ColorB{}{w}[c] \wedge \VRankB{w}[r]}
  \end{align*}
  \end{definition}
  Note that for every $P$ the set $\{\Pair{c}{r}: P \text{ is in upward closure of } \{\FunAugment{}{c}{r}\}\}$ can be
  enumerated in space $\mathcal{O}(\log(\NR)\Param)$. 
  Note also that $\{P: P \text{ is in upward closure of } \{\FunAugment{}{c}{r}\}\}$ is a valid enforcement set for 
  a vertex $v$ in one vertex game $\G$ where $v$ has rank $r$ and color $c$.

\subsubsection{Combination Gadget}
We define a constant-depth gadget that given $m$ packs of input bits, each of length $l$
($i$'th bit of $j$'th pack denoted as $\CIn{j}[i]$),
and $m$ bits that encode the choice of the bit pack ($j$'th bit denoted as $\CCond{}[j]$),
outputs the result of the or operation on chosen bit packs.
\begin{definition}{$\CChoose{m}{l}$}
    \[\COut{}[i] = \bigvee_{1 \leq j \leq m}{\CIn{j}[i] \wedge \CCond{}[j]}\]
\end{definition}

\subsubsection{Update Gadget}\label{sec:update_gadgets}
Here we define a constant depth gadget that given for each vertex $v$ of a subgame $\G'$ of $\G$ a valid set of enforcements
outputs a valid set of enforcements for each vertex of the game $\G''$ created from $\G'$ by adding directed edges
from vertices of color $s$ to those of color $t$.

Let $\ParityGameDefC{\G'}{\V'}{\E'}{\R'}{\Color'}$ be any subgame of $\G$.
For a pair of colors $\Pair{\CFrom}{\CTo}$ we define a gadget $\UpdateGadget{\CFrom}{\CTo}$ that inputs
one bit $\CIsEColor{}[\CFrom]$,
and  for each vertex $v$ of $\G$ one bit encoding whether $v \in \V'$ (denoted as $\VertexUsedBit{\IN}[v]$), $\ColorB{}{v}$, and
$\ACRankBits$ bits denoted as $\EnfB{\IN}{v}$. For a vertex $v$ of $\G'$,
$\EnfB{\IN}{v}$ will encode a set of enforcements that is valid for $v$ in $\G'$ 
(one bit for one partial function from $\C$ to $\{0, \ldots, d-1\}$; bit for enforcement $P$ denoted as $\EnfB{\IN}{v}[P]$).
We define $\UpdateGadget{\CFrom}{\CTo}$'s output 
to be $\ACRankBits$ bits for each vertex $v$ of $\G$ (denoted as $\EnfB{\OUT}{v}$). 
For a vertex $v \in \V'$ those bits will encode enforcement set that is valid for $v$ in the game $\G''$
created from $\G'$ by adding connections from vertices colored $\CFrom$ to ones colored $\CTo$.

We consider two cases for a pair of colors $\Pair{\CFrom}{\CTo}$.
The first one is when color $\CFrom$ belongs to player $\PE$ (i.e. $\CIsEColor{}[\CFrom]$ is $\True$) and the 
other is when it belongs to player $\PO$. For the first case,
we use gadget $\GUpdateEV{\CFrom}{\CTo}$, whereas for the second one $\GUpdateOV{\CFrom}{\CTo}$,
both of which are defined below.

\begin{definition}{$\GUpdateEV{\CFrom}{\CTo}$}\\
  \begin{align*}
  \EnfOneMove{}[P] &=  \bigvee_{w}{\VertexUsedBit{\IN}[w] \wedge \ColorB{}{w}[\CTo] \wedge \EnfB{\IN}{w}[P]}\\
  \EnfCycle{}[P] &=  \NoPairWith{P}{\CFrom} \wedge \bigvee_{2r \in \{0, \ldots, \NR - 1\}}{\EnfOneMove{}[\FunAugment{P}{\CFrom}{2r}]}\\
  \EnfOpt{}[P] &= \EnfOneMove{}[P] \vee \EnfCycle{}[P]\\
  \EnfB{\OUT}{v}[P] &= \VertexUsedBit{\IN}[v] \wedge  \left(\bigvee_{S \ECMPEQ P}{\EnfB{\IN}{v}[S]} \vee \bigvee_{\Pair{Q}{R} : \EnfMerge{}{Q}{\CFrom}{R} \ECMPEQ P}{\EnfB{\IN}{v}[Q] \wedge \EnfOpt{}[R]}\right)
  \end{align*}
  Where $\EnfMerge{}{Q}{\CFrom}{R}$ is calculated as if $\CFrom$ was the color of player $\PE$.
\end{definition}

The intuition behind $\GUpdateEV{\CFrom}{\CTo}$ is that we should be able to modify each positional strategy $\ST_\PE$ that is safe from $v$
in game $\G''$ to a strategy $\ST'_\PE$ that behaves similarly to $\ST_\PE$ but when $\ST_\PE$ decides to make a move along a new edge, then $\ST'_\PE$ always chooses to move to a fixed vertex $w$ of color $t$.
What is more, we should be able to make this modification in a way that $\EnforcementC{\G''}{\TreeModelName}{\ST'_\PE}{v} \ECMPEQ \EnforcementC{\G''}{\TreeModelName}{\ST_\PE}{v}$ and $\ST'_\PE$ is also safe from $v$. 
In $\EnfB{\OUT}{v}[P]$ we check for the existence of such $\ST'_\PE$ with $\EnforcementC{\G''}{\TreeModelName}{\ST'_\PE}{v} \ECMPEQ P$.
We do it by splitting $\ST'_\PE$ into two parts: part till the new move (with enforcement $\ECMPEQ Q$) and after (with enforcement $\ECMPEQ R$).
Note that $\ST'_\PE$ defined as above will either make a move along a new edge at most once ($\EnfOneMove{}$) or an unbounded number of times ($\EnfCycle{}$).

\begin{definition}{$\GUpdateOV{\CFrom}{\CTo}$}\\
  First, for each pair of enforcements $Q$, $P$,  and a vertex $v$ such that
  $Q \ECMPEQ P$:
  \begin{align*}
  \AnyGoodEnfForQPv{Q}{P}{v} &= \neg \VertexUsedBit{\IN}[v] \vee \neg \ColorB{}{v}[\CTo] \vee \bigvee_{R : \EvenPair{R}{\CFrom} \wedge \EnfMerge{}{Q}{\CFrom}{R} \ECMPEQ P}{\EnfB{\IN}{v}[R]}\\
  \AnyGoodEnfForQP{Q}{P} &= \bigwedge_{w}{\AnyGoodEnfForQPv{Q}{P}{w}}
  \end{align*}
  Where $\EvenPair{R}{s}$ is $\True$ if and only if $R(s)$ is either even or \UndefinedT{}, and
  $\EnfMerge{}{Q}{\CFrom}{R}$ is calculated as if $\CFrom$ was the color of player $\PO$.\\
  \begin{align*}
  \EnfB{\OUT}{v}[P] &= \VertexUsedBit{\IN}[v] \wedge \bigvee_{Q \ECMPEQ P}{\EnfB{\IN}{v}[Q] \wedge (\NoPairWith{Q}{\CFrom} \vee \AnyGoodEnfForQP{Q}{P})}
  \end{align*}
\end{definition}

The intuition behind $\GUpdateOV{\CFrom}{\CTo}$ is that each safe positional strategy in game $\G''$
should factor to a strategy up to reaching a vertex of color $s$ and a bunch of strategies after the move of player $\PO$.
The part $\AnyGoodEnfForQPv{Q}{P}{v}$ corresponds to checking whether given a safe strategy with enforcement $\ECMPEQ Q$ in game $\G'$
we can make it respond to a move of player $\PO$ to a vertex $v$ in a way that maintains the safety of the strategy and makes its enforcement non-worse than $P$.

We merge auxiliary gadgets using $\CChoose{2}{n\ACRankBits}$ with input packs being the outputs of 
$\GUpdateEV{\CFrom}{\CTo}$ and $\GUpdateOV{\CFrom}{\CTo}$ respectively and $\CCond{}[1] = \CIsEColor{}[\CFrom]$, 
$\CCond{}[2] = \neg\CIsEColor{}[\CFrom]$ ($\GUpdateEV{\CFrom}{\CTo}$ was created to handle the case when $\CFrom$ is the color of player $\PE$
and $\GUpdateOV{\CFrom}{\CTo}$ for the other one, here we use $\CChoose{2}{n\ACRankBits}$ to choose between their results).

Note that by transitivity of $\ECMPEQ$, we have that sets $\EnfB{\OUT}{v}$ produced by $\GUpdateEV{\CFrom}{\CTo}$
and $\GUpdateOV{\CFrom}{\CTo}$ are up-closed.
Note also that $\UpdateGadget{\CFrom}{\CTo}$  has constant depth and can be computed in space $\AlgComplexitySpace$.
This is because $\mathcal{O}{(\log(n))}$ bits suffice to remember a constant number of vertices. 
What is more, as each enforcement is a partial function from the set of $k$ colors to the set of $d$ ranks, so to encode one
$\mathcal{O}(\log(d)k)$ bits suffice, and we consider at most four of them simultaneously.

\begin{restatable}{lemma}{LemmaUpdateValid}\label{lemma:update_valid}
  Given that $\UpdateGadget{\CFrom}{\CTo}$ inputs a valid set of enforcements for 
  each vertex $v$ of a subgame $\G'$ of game $\G$ then it outputs a valid set of enforcements for 
  each vertex $v$ of a game $\G''$ created from $\G'$ by addition of directed edge from every vertex of
  color $\CFrom$ to every vertex of color $\CTo$.
\end{restatable}

A simple case analysis suffices to prove this lemma when $\CFrom$ belongs to player $\PO$.
When $\CFrom$ belongs to player $\PE$ one can separately prove soundness (again by simple case analysis)
and completeness following the presented intuition.
As these proofs are rather long and tedious
they have been postponed to \cref{sec:ski_lem_uev}.

\subsubsection{Conditional Update Gadget}
For each pair of colors $\Pair{\CFrom}{\CTo}$ we  define $\UpdateGadgetC{\CFrom}{\CTo}$ 
that inputs the bits that $\UpdateGadget{\CFrom}{\CTo}$ inputs with one additional bit 
that is $\True$ if and only if the update defined by $\UpdateGadget{\CFrom}{\CTo}$
should be performed (i.e. whether enforcements should be updated or copied). 
We define $\UpdateGadgetC{\CFrom}{\CTo}$  to be a simple combination 
of $\UpdateGadget{\CFrom}{\CTo}$ and $\CChoose{2}{n\ACRankBits}$.

\subsubsection{Processing Gadget}\label{sss:proc_gadget}
Here we define a gadget of depth $\mathcal{O}(\Param^2)$ that stacks operations of conditional update gadgets.

Let $p_1, \ldots, p_{\Param^2}$ be pairs of colors in some fixed order.
We define the gadget $\ProcPipe$ as a sequence of gadgets $\UpdateGadgetCP{p_1}, \ldots \UpdateGadgetCP{p_{\Param^2}}$
where $\EnfB{\IN}{v}$ part of the input for gadget $\UpdateGadgetCP{p_i}$ comes from the output of $\UpdateGadgetCP{p_{i-1}}$
for $i > 1$ and rest of the input of $\UpdateGadgetCP{p_j}$ for $j > 0$ comes from the input 
of $\ProcPipe$. In particular for $\UpdateGadgetCP{p_j}$ the bit that tells whether to perform the update comes from
the pack of $\Param^2$ bits (denoted as $\TCSBName^\IN$) that $\ProcPipe$ inputs aside from $\CIsEColor{}$ and 
$\ColorB{}{v}$, $\EnfB{\IN}{v}$, $\VertexUsedBit{\IN}[v]$ for each $v$. We define the output of $\ProcPipe$ to be the output of $\UpdateGadgetCP{p_{\Param^2}}$
and $\VertexUsedBit{\OUT}[v] = \VertexUsedBit{\IN}[v]$ for each $v$.

Note that $\ProcPipe$  defined above has depth $\mathcal{O}(\Param^2)$ and can be computed in space 
$\AlgComplexitySpace$, as each of the conditional update gadgets can be computed in this space, and we can 
use $\mathcal{O}(\log(\Param))$ bits to encode for which $\UpdateGadgetCP{p_j}$ we are performing the computation.

Let $l$ be an internal node of $\T$ from the given tree-model and let $l_1, \ldots, l_m$ be children of $l$. Gadget $\ProcPipe$ will be later
used to calculate valid enforcement sets for a subgame of $\G$ played on an arena induced by $\TreeModelName_l$ 
from valid enforcement sets for a subgame of $\G$ played on an union of arenas induced by $\TreeModelName_{l_1}, \ldots, \TreeModelName_{l_m}$.

\subsection{Circuit Construction}\label{ss:circ_constr}
The construction of the circuit will proceed in layers, with each layer output being $n$ big packs of bits,
each containing $\EnfB{\OUT}{v}$ and $\VertexUsedBit{\OUT}[v]$ for each vertex $v$, so $n(n(\ACSBitPack + 1))$ bits in total.
The output of each layer will be treated as the input for the next one.
The output of $l$'th layer will be denoted as $\CLayerPack{l}$.
To start the construction off and create layer $0$, we use $\CInitialize{v}$ gadget for each vertex $v$.
To create parts for layer $1 \leq l \leq \Param$
we need to use $\ProcPipe$ for each node $\Pair{l}{i}$. 
We do so by setting $\ProcPipe$'s $\TCSBName^\IN$ bits to $\TCSB{l}{i}$,
whereas  $\EnfB{\IN}{v}$ and $\VertexUsedBit{\IN}[v]$ bits using $\CChoose{n}{n(\ACSBitPack + 1)}$ 
with input being $n$ big bit packs from $\CLayerPack{l - 1}$ and $\CCond[j] = \TCCB{l}{i}[j]$
(here $\CChoose{n}{n(\ACSBitPack + 1)}$ is used to take a disjoint union of arenas).
The output of $\ProcPipe$ here constitutes to $\CLayerPack{l}$.

To produce partition of $\V$ into winning regions
 we first gather calculated results using $\CChoose{n}{n(\ACSBitPack + 1)}$ gadget with input being set to
$n$ big bit packs from $\CLayerPack{\Param}$ and $\CCond[j] = \True$, and then 
create for each vertex $v$ gadget that checks whether there is an enforcement $P$
such that $\EnfB{\IN}{v}[P]$ is $\True$ and for each $c \in \Domain{P}$ 
$\CIsEColor{}[c]$ is $\False$.

\begin{theorem}\label{theorem:ac_size_and_complex}
  There exists a Turing machine that given $n, \NR > 1$, and $\Param > 0$
  computes in space $\AlgComplexitySpace$
  circuit $\Circuit_{\Tuple{n, \NR, \Param}}$ of depth $\ACAlgDepth$,
   that inputs 
  parity game $\ParityGameDefC{\G}{\V}{\E}{\R}{\Color}$ and tree-model $\TreeModelShort$
  such that $\R(V) \subseteq \{0, \ldots, \NR - 1\}$, $\NR \leq 2n$, $|\C| = k$, $\T$ has height exactly $k$ and $\TreeModelName$ induces $\G$'s arena,
  and outputs
   the partition of $\G$'s vertices into winning regions of 
  players $\PE$ and $\PO$.
\end{theorem}

\begin{proof}
  The complexity follows from the observations made during the construction.
  Correctness follows from \cref{lemma:update_valid} and simple induction on the size of the circuit.

\end{proof}

\section{Circuit Family for Tree-Depth}\label{sec:circ_td}
Now we will show how for $n, \NR > 1$, and $\Param \geq 0$ compute in space 
$\AlgComplexitySpace$ circuit $\TDCircuit_{\Tuple{n, \NR, \Param}}$
of depth being at most $\ACAlgDepth$, that
given an $n$-vertex parity game with all ranks being less that $\NR$ (we also require $\NR \leq 2n$) 
and elimination forest of height at most $\Param$, outputs the winner for each vertex.
We achieve this by creating circuit $\TRCircuit_{\Tuple{n, \NR, \Param}}$ that changes input representation
to the one accepted by $\CFTD{}$, and then getting the winners
from the output of $\CFTD{}$
using circuit $\TGCircuit_{\Tuple{n, \NR, \Param}}$.
Note that along with the construction of the circuit $\TRCircuit_{\Tuple{n, \NR, \Param}}$ 
we show a more efficient reduction that does not have an exponential blow-up
in $k$ (in contrast to the simple one presented in \cref{ssec:td_intro}). 
Such improvement is not necessary for the main result of this paper, but it makes
the construction more computationally tractable.
In the end, we show how to get rid of the requirement for the elimination forest.

\subsection{Circuit Input}

Circuit $\TRCircuit_{\Tuple{n, \NR, \Param}}$ inputs a game $\ParityGameFullDef{\G}{\V}{\E}{\R}$
and elimination forest $\Forest$ as the following sequence of bits.
\begin{itemize}
    \item For each vertex $v \in \V$:
    \begin{itemize}
        \item $\NR$ bits that one-hot encode $\R(v)$
        (denoted as  $\VRankB{v}$),
        \item one bit that is true if and only if $v$ is a vertex
        of player $\PE$ (denoted as $\EPB{v}$),
        \item $\Param + 1$ bits that one-hot encode the depth of $v$ in $\Forest$
        (denoted as $\DepthB{v}$; bit $\DepthB{v}[x]$ is set to true if $v$ has depth $x$; root has depth $0$),
        \item  $n - 1$ bits that encode the parent of $v$ in $\Forest$
        (denoted as $\ParentB{v}$; $\ParentB{v}[w]$ is $\True$ if $w$ is a parent of $v$ in $\Forest$;
        each non-root vertex has exactly one parent).
    \end{itemize}
    \item For each pair of vertices $v, w \in V$, a bit that is set to true if and only if 
    there is a directed edge from $v$ to $w$ (denoted as $\VWEdgB{v}{w}$).
\end{itemize}
Checking whether $\Forest$ is a valid forest can be done similarly as in \cref{sec:ac0_shrub_circ_in}.
For ancestor-descendant relationship one can compute in space $\AlgComplexitySpaceL$ a gadget of depth $\mathcal{O}(\log(\Param))$
that at level $i$ checks for every $v, w \in \V$ whether there is a simple path from $v$ to $w$ that goes towards the root
and has length at most $2^i$.

\subsection{Reduction Circuit}
First of all, we do not aim to create exactly the same game as input for $\CFTD$,
but a modified game from which we could deduce winning regions in the original one.
The main idea here is to split each vertex in a number of vertices such that
when we merge them back, then we get the original arena.
Let $\V = \{v_1, \ldots, v_n\}$ be all vertices of $\G$, with each vertex $v \in \V$
we associate $2n$ copies of $v$ with the same player ownership: 
$ v \FT \CVB{v}{v_1}, \ldots, \CVB{v}{v_n}, \CVEB{v}{v_1}, \ldots, \CVEB{v}{v_n}$.
 The first $n$ of these copies
are denoted as $\CVB{v}{}$, whereas the last $n$ as $\CVEB{v}{}$. 
Each of them corresponds to one additional vertex from the input of $\TRCircuit_{\Tuple{n, \NR, \Param}}$
(the copy among the first $n$ copies of $v$ that corresponds to $w$ is denoted as $\CVB{v}{w}$,
whereas the copy among the last $n$ copies that corresponds to $w$ is denoted as $\CVEB{v}{w}$).
Intuitively for a vertex $v \in \V$, copies from $\CVB{v}{}$ and $\CVEB{v}{}$ 
will be used to distribute the original connections from and to $v$. More precisely, 
vertices from $\CVB{v}{}$ will allow to choose a connection from $v$, whereas vertices from $\CVEB{v}{}$
to execute it, and if there will be a connection from $v$ to $w$ in $\G$, then 
we will create a connection from some $\CVEB{v}{v'}$ to some $\CVB{w}{w'}$.

We set the rank of
copies from $\CVEB{v}{}$ to $\R(v) + 2$, whereas
the rank of copies from $\CVB{v}{}$ to $1$ if $v$ belongs to player $\PE$, and to 
$0$ if $v$ belongs to player $\PO$ (as $\CVB{v}{}$ will be used to choose a connection we don't want 
a player to put off the decision indefinitely).
Let $x$ be the depth of $v$ in $\Forest$.
We color copies $\CVB{v}{}$ using color $\Tuple{\TDPCCVVC, \EPB{v}, x}$,
whereas copies from $\CVEB{v}{}$ using $\Tuple{\TDPCCVEC, \EPB{v}, x}$.
The copies created here are leaves of the tree-model that $\CFTD$ will take as input. 
Note that for all operations presented above we can compute in space $\AlgComplexitySpaceL$ gadget of constant depth,
as we need $\mathcal{O}(\log(n))$ bits to represent a copy of a vertex and rank and $\mathcal{O}(\log(k))$ to represent a color.

\medskip

Before we go any further, we introduce some gadgets. For simplicity, we will assume that 
vertex is its own ancestor.
\begin{itemize}
  \item $\CNoChildren{v}$ -- that is $\True$ if and only if $v$ has no children in $\Forest$.
  \item $\CAncestorEQ{v}{w}$ -- that is $\True$ if and only if $w$ is an ancestor of $v$ in $\Forest$.
  Note that it can be done by a gadget with $\mathcal{O}(\log(k))$ layers computable in space $\AlgComplexitySpaceL$ that in layer 
  $i$ considers paths from vertices towards root of length at most $2^i$.
  \item $\CAncestorEQEBTW{v}{x}{x'}$ --
   that is true if and only there are two ancestor $w$ and $w'$ of $v$ of depth $x$ and $x'$ respectively such that there is an an edge from $w$ to $w'$.
\end{itemize}

Now we will describe how we compute the part of $\TRCircuit_{\Tuple{n, \NR, \Param}}$ that outputs the rest of the tree-model.
The main thing that we will do here is that for each vertex $v$ that is a leaf in the given 
elimination forest we will gather copies of ancestors that correspond to $v$, add connections using the copies and join
the copies along the way creating a tree-model similar in shape to the given elimination forest.
The computation will consist of $\Param + 1$ major steps, and $3\Param + 3$ simple steps. 
Major steps will create layers of the tree-model (see \cref{sec:ac0_shrub_circ_in}; layers are numbered starting from the layer of leaves, which has number $0$), whereas simple steps will extend the height of the tree-model to
match the specification for  $\CFTD{}$. 
As in the description of the tree-model for $\CFTD{}$ we use graph where number of nodes in every layer is the same and in our case equals number of copies of vertices,
so we will use copies of vertices to number nodes from a specific layer. The tree-model that we 
will produce here may not use all of the leaves
(there may be vertices at the bottom that will be left isolated), 
we get around it by relaxing the correctness of $\CFTD{}$ input
so that it will allow vertices from layer $0$ without parents (note that they will be ignored in the further parts of the circuit).
As the steps presented below are rather technical we encourage to refer to \cref{fig:major_steps}.

First the major steps, for $l = 1$, for each vertex $v$ of the input game that is childless in $\Forest$,
we gather and initialize copies of ancestors of $v$ corresponding to $v$.
That is for each $v, w \in \V$, $p \in \{\True, \False\}$ and $x \in \{0, \ldots, \Param\}$ we set
\begin{align*}
  \TCCB{l}{\CVB{v}{v}}[\CVB{w}{v}] &= \CAncestorEQ{v}{w} \wedge \CNoChildren{v}\\
  \TCCB{l}{\CVB{v}{v}}[\CVEB{w}{v}] &= \CAncestorEQ{v}{w} \wedge \CNoChildren{v}\\
  \TCSB{l}{\CVB{v}{v}}[\Pair{\Tuple{\TDPCCVVC, p, x}}{\Tuple{\TDPCCVEC, p, x}}] &= \CNoChildren{v}
\end{align*}
Then we add all connections between gathered copies, that is for
$v \in \V$, $p, p' \in \{\True, \False\}$ and $x, x' \in \{0, \ldots, \Param\}$ we set
\begin{align*}
  \TCSB{l}{\CVB{v}{v}}[\Pair{\Tuple{\TDPCCVEC, p, x}}{\Tuple{\TDPCCVVC, p', x'}}] &= \CNoChildren{v} \wedge \CAncestorEQEBTW{v}{x}{x'}
\end{align*}

Now let $1 < l \leq \Param + 1$. 
Here we create parts that drag low depth vertices up in the tree-model and merge gathered copies.
That is for $p, p' \in \{\True, \False\}$, $x \in \{0, \ldots, \Param - l + 1\}$, $w\not=v$ we set:
\begin{align*}
  \TCCB{l}{\CVB{v}{v}}[\CVB{v}{v}] &= \CNoChildren{v} \wedge \bigvee_{t \in \{0, \ldots, \Param - l + 1\}}{\DepthB{v}[t]}\\
  \TCCB{l}{\CVB{v}{v}}[\CVB{w}{w}] &= \DepthB{v}[\Param - l + 1] \wedge \ParentB{w}[v]\\
  \TCSB{l}{\CVB{v}{v}}[\Pair{\Tuple{\TDPCCVVC, p, \Param - l + 1}}{\Tuple{\TDPCCVVC, p, \Param - l + 1}}] &= \DepthB{v}[\Param - l + 1]
\end{align*}

For simple steps, we extend the height of the tree-model to match specification for $\CFTD$.
Let $v$ be some fixed vertex in $\V$ and $w$ any vertex from $\V$, and let $\Param + 3 \leq l \leq 4\Param +4$,  we set
$\TCCB{\Param + 2}{\CVB{v}{v}}[\CVB{w}{w}]$ to $\DepthB{w}[0]$ whereas $\TCCB{l}{\CVB{v}{v}}[\CVB{v}{v}]$ to $\True$.

In both major and simple steps, we set unmentioned bits to $\False$.

\begin{figure}
  
  \centering

    \begin{subfigure}{.5\textwidth}
      \centering
      \begin{tikzpicture}[>=stealth', line width=1pt, minimum size=0.8cm]
      \node[draw, circle] at (0, 1.2) (a) {$a$};
      \node[draw, circle] [below = 0.5 of a] (b) {$b$};
      \node[draw, circle] [below left = 0.5 and 0.5 of b] (c) {$c$};
      \node[draw, circle] [below right = 0.5 and 0.5 of b] (d) {$d$};
  
      \draw[->] (a) -- (d);
      \draw[->] (c) -- (a);
      \draw[->] (b) -- (c);
      \end{tikzpicture}
      \subcaption{Arena}
    \end{subfigure}%
    \begin{subfigure}{.5\textwidth}
      \centering
      \begin{tikzpicture}[>=stealth', line width=1pt, minimum size=0.8cm]
        \node[draw, circle] at (5, 1.2) (a) {$a$};
        \node[draw, circle] [below left = 0.5 and 0.5 of a] (b) {$b$};
        \node[draw, circle] [below = 0.5 of b] (c) {$c$};
        \node[draw, circle] [below right = 0.5 and 0.5 of a] (d) {$d$};
  
        \draw (b) -- (a);
        \draw (c) -- (b);
        \draw (d) -- (a);
      \end{tikzpicture}
      \subcaption{Elimination forest}
    \end{subfigure}
   
    \begin{subfigure}{1.0\textwidth}
      \centering
      \vspace{-3cm}
      \begin{tikzpicture}[>=stealth', line width=1pt, minimum size=0.8cm]
      \phantom{
        \TMLV{}
        \TMLIV{}
        \TMLIII{}}
        \TMLI{}
        \TML{}
        \phantom{
        \TMLEII{}
        \TMLEIV{}
        \TMLEV{}}

        \path[->] (ac) edge [out=280, in=250, looseness=2] (aec);
        \path[->] (bc) edge [out=280, in=250, looseness=2] (bec);
        \path[->] (cc) edge [out=280, in=250, looseness=2] (cec);

        \path[->] (ad) edge [out=280, in=250, looseness=2] (aed);
        \path[->] (dd) edge [out=280, in=250, looseness=2] (ded);

        \path[->] (cec) edge [out=280, in=250, looseness=0.6] (ac);
        \path[->] (bec) edge [out=280, in=250, looseness=2] (cc);
        \path[->] (aed) edge [out=280, in=250, looseness=1.5] (dd);

      \end{tikzpicture}
      \subcaption{Layers $0$ and $1$ from tree-model  and connections added at major step $1$ (bottom).}
    \end{subfigure}

    \begin{subfigure}{1.0\textwidth}
      \centering
      \begin{tikzpicture}[>=stealth', line width=1pt, minimum size=0.8cm]
  
          \TMLV{0}  
        \TMLIV{}
        \TMLIII{}
        \TMLI{}
        \TML{}
        \TMLEII{}
        \TMLEIV{}
        \TMLEV{}
  
        \path[<->] (ad) edge [out=280, in=250, looseness=0.5] (ac);
        \draw[->] (ac) to [out=260,in=284,looseness=6]  (ac);
        \draw[->] (ad) to [out=270,in=250,looseness=6]  (ad);
      \end{tikzpicture}
      \subcaption{Layers $0,1,2,3,4$ from tree-model and merge performed at major step $4$ (bottom).}
    \end{subfigure}
    \caption{Tree-model produced for a simple arena (a) with elimination forest (b).}\label{fig:major_steps}

  \end{figure}
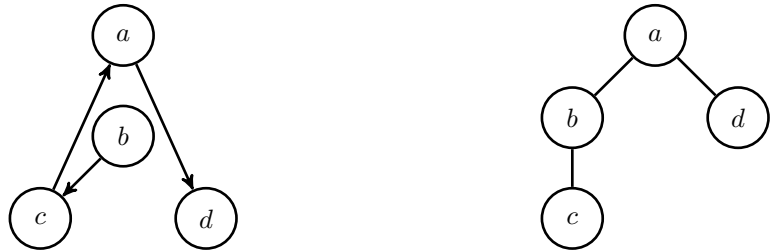
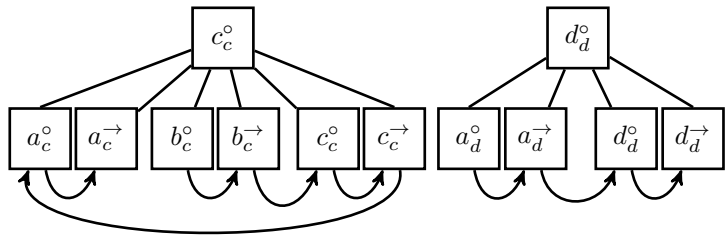
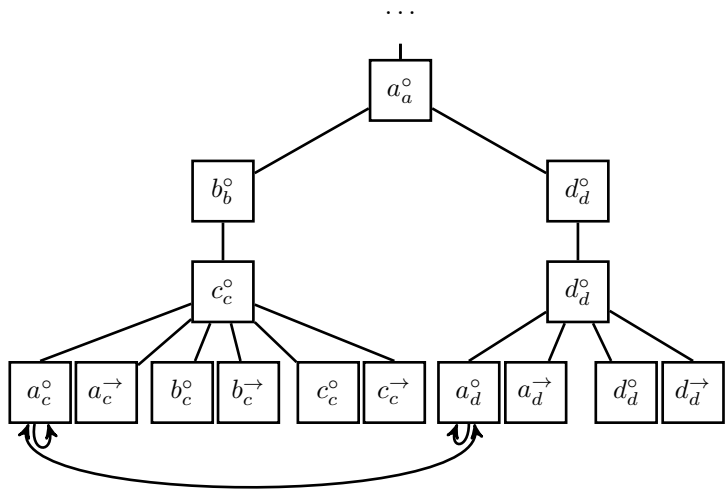

\begin{proposition}\label{proposition:efficient_reduction}
The depth of $\TRCircuit_{\Tuple{n, \NR, \Param}}$ when constructed as above is $\mathcal{O}(\log(\Param))$, what is 
more the construction can be done in space $\AlgComplexitySpaceL$.
\end{proposition}

\begin{proof}
    Observe that $\TCCB{l}{\CVB{v}{v}}$ and $\TCSB{l}{\CVB{v}{v}}$ bits depend only on bits describing depth and ancestors.
    What is more, $\AlgComplexitySpaceL$ bits suffice to describe which step $l$ we are performing
    and for which bit of either $\TCCB{l}{\CVB{v}{v}}$ or $\TCSB{l}{\CVB{v}{v}}$
    we are creating a gate.
\end{proof}

\subsection{Output Circuit}
Now we will proceed with the definition of  $\TGCircuit_{\Tuple{n, \NR, \Param}}$.
To get the winning region of $\PE$ we need to check 
for each vertex $v \in \V_\PE$ whether there is a copy of $\CVB{v}{w}$ such that $\PE$ wins from  
$\CVB{v}{w}$ and for each $v \in \V_\PO$ whether $\PE$ wins from all copies of $v$.
It suffices as note that we have constructed input for $\CFTD{}$ such that
when we contract all $\CVB{v}{}$ and $\CVEB{v}{}$ into one vertex of rank $\R(v)$ then we get the original game.
\[\COut{\TGCircuit{}}[v] = (\EPB{v} \wedge \bigvee_{w \in \V}{\COut{}[\CVB{v}{w}]}) \vee (\neg \EPB{v} \wedge \bigwedge_{w \in \V}{\COut{}[\CVB{v}{w}]}) \]
where $\COut{}$ denotes the output bits of $\CFTD{}$. Note that $\TGCircuit_{\Tuple{n, \NR, \Param}}$ can be computed in space
$\mathcal{O}(\log(n))$ as  $\mathcal{O}(\log(n))$ bits suffice to handle a constant number of vertices.

\begin{theorem}\label{thm:td_given_solve_log}
  There exists a Turing machine that given
    $n, \NR > 1$, and $\Param \geq 0$ computes in space $\AlgComplexitySpace$ circuit $\TDCircuit_{\Tuple{n, \NR, \Param}}$ of depth $\ACAlgDepth$,
    that inputs $n$-vertex parity game $\G$ 
    along with elimination forest of $\G$'s arena, and outputs the partition of $\G$'s vertices into winning regions of 
    players $\PE$ and $\PO$.
\end{theorem}
\begin{proof}
    Follows from \cref{proposition:efficient_reduction} and \cref{theorem:ac_size_and_complex}.
\end{proof}

\subsection{Computing Elimination Forest}\label{sec:EX}

The following is an easy consequence of the result from \cite[Theorem 1]{elberfeld:hal-00678176}.
\begin{lemma}\label{theorem:treedepth_circ}
  For any $\Param \geq 0$ there exists a logspace uniform family of $\ACO$-circuits 
  that inputs a graph and checks whether its tree-depth is $\Param$.
\end{lemma}

Now observe that we can get an elimination forest for an arena of tree-depth at most $\Param$
by a circuit that forgets about the orientation of edges and performs $\Param$ steps.
At step $i \in \{1, \ldots, \Param\}$ it considers a graph of tree-depth at most $\Param - i + 1$ and checks
which vertices should be removed to make it a graph of tree-depth at most $\Param-i$. 
It does so by taking for each vertex $v$ its connected component,
removing $v$ from it and
checking the tree-depth of the component with $v$ removed using a circuit from \cref{theorem:treedepth_circ}.
In case there are many candidates the circuit picks the first one.
From \cref{theorem:treedepth_circ} and \cref{lemma:treedepth_short_paths},
for fixed $\Param$ this circuit can have constant depth and be computed in logarithmic space
(to compute connected components it suffices to check for paths of length at most $2^{k+1} - 2$
what can be done similarly as in $\CAncestorEQ{v}{w}$). 
So from the construction above and \cref{thm:td_given_solve_log} we have the following.

\begin{corollary}
  For any $\Param \geq 0$ there exists a logspace uniform family of $\ACO$-circuits 
  that solves parity games played on arenas of tree-depth at most $\Param$.
\end{corollary}

\section{Conclusions and Further Work}
We have shown that solving parity games played on arenas of bounded tree-depth or shallow clique-width is in logspace uniform $\ACO$, 
assuming that the suitable tree-model is provided for the latter case.
Speaking of uniformity, one can 
assign appropriate binary identifiers to gates such that for any $\Param$, the problem: 
``Given circuit parameters $n$, $\NR$ in binary and binary identifiers of two gates $A$, $B$ decide whether in  $\TDCircuit_{\Tuple{n, \NR, \Param}}$ both $A$ and $B$ are present, and
 the output of $B$ constitutes to the 
input of $A$''
can be solved by a deterministic Turing machine that works in time polylogarithmic in $n$ and uses space logarithmic in $n$. 
However, going lower seems to require more technical care. 
A natural next step in research would be
to decide whether one can obtain a similar result for the bounded shrub-depth case.
However, as we observed in \cref{sec:circ_overview_brief}, such step seems to require a different approach.

\bibliography{refs}

\begin{thebibliography}{10}

\bibitem{Ajtai198311FormulaeOF}
Mikl{\'{o}}s Ajtai.
\newblock {\(\sum\)}\({}^{\mbox{1}}\)\({}_{\mbox{1}}\)-formulae on finite
  structures.
\newblock {\em Ann. Pure Appl. Log.}, 24(1):1--48, 1983.
\newblock \href {https://doi.org/10.1016/0168-0072(83)90038-6}
  {\path{doi:10.1016/0168-0072(83)90038-6}}.

\bibitem{arora_barak_2009}
Sanjeev Arora and Boaz Barak.
\newblock {\em Computational Complexity - {A} Modern Approach}.
\newblock Cambridge University Press, 2009.
\newblock URL:
  \url{http://www.cambridge.org/catalogue/catalogue.asp?isbn=9780521424264}.

\bibitem{10.1007/11672142_43}
Dietmar Berwanger, Anuj Dawar, Paul Hunter, and Stephan Kreutzer.
\newblock Dag-width and parity games.
\newblock In Bruno Durand and Wolfgang Thomas, editors, {\em {STACS} 2006, 23rd
  Annual Symposium on Theoretical Aspects of Computer Science, Marseille,
  France, February 23-25, 2006, Proceedings}, volume 3884 of {\em Lecture Notes
  in Computer Science}, pages 524--536. Springer, 2006.
\newblock \href {https://doi.org/10.1007/11672142\_43}
  {\path{doi:10.1007/11672142\_43}}.

\bibitem{10.1007/978-3-540-32275-7_15}
Dietmar Berwanger and Erich Gr{\"{a}}del.
\newblock Entanglement - {A} measure for the complexity of directed graphs with
  applications to logic and games.
\newblock In Franz Baader and Andrei Voronkov, editors, {\em Logic for
  Programming, Artificial Intelligence, and Reasoning, 11th International
  Conference, {LPAR} 2004, Montevideo, Uruguay, March 14-18, 2005,
  Proceedings}, volume 3452 of {\em Lecture Notes in Computer Science}, pages
  209--223. Springer, 2004.
\newblock \href {https://doi.org/10.1007/978-3-540-32275-7\_15}
  {\path{doi:10.1007/978-3-540-32275-7\_15}}.

\bibitem{10.1145/3055399.3055409}
Cristian~S. Calude, Sanjay Jain, Bakhadyr Khoussainov, Wei Li, and Frank
  Stephan.
\newblock Deciding parity games in quasipolynomial time.
\newblock In Hamed Hatami, Pierre McKenzie, and Valerie King, editors, {\em
  Proceedings of the 49th Annual {ACM} {SIGACT} Symposium on Theory of
  Computing, {STOC} 2017, Montreal, QC, Canada, June 19-23, 2017}, pages
  252--263. {ACM}, 2017.
\newblock \href {https://doi.org/10.1145/3055399.3055409}
  {\path{doi:10.1145/3055399.3055409}}.

\bibitem{COURCELLE1993218}
Bruno Courcelle, Joost Engelfriet, and Grzegorz Rozenberg.
\newblock Handle-rewriting hypergraph grammars.
\newblock {\em J. Comput. Syst. Sci.}, 46(2):218--270, 1993.
\newblock \href {https://doi.org/10.1016/0022-0000(93)90004-G}
  {\path{doi:10.1016/0022-0000(93)90004-G}}.

\bibitem{COURCELLE200077}
Bruno Courcelle and Stephan Olariu.
\newblock Upper bounds to the clique width of graphs.
\newblock {\em Discret. Appl. Math.}, 101(1-3):77--114, 2000.
\newblock \href {https://doi.org/10.1016/S0166-218X(99)00184-5}
  {\path{doi:10.1016/S0166-218X(99)00184-5}}.

\bibitem{10.1007/978-3-540-87531-4_26}
Anuj Dawar and Erich Gr{\"{a}}del.
\newblock The descriptive complexity of parity games.
\newblock In Michael Kaminski and Simone Martini, editors, {\em Computer
  Science Logic, 22nd International Workshop, {CSL} 2008, 17th Annual
  Conference of the EACSL, Bertinoro, Italy, September 16-19, 2008.
  Proceedings}, volume 5213 of {\em Lecture Notes in Computer Science}, pages
  354--368. Springer, 2008.
\newblock \href {https://doi.org/10.1007/978-3-540-87531-4\_26}
  {\path{doi:10.1007/978-3-540-87531-4\_26}}.

\bibitem{10.5555/3134208}
Reinhard Diestel.
\newblock {\em Graph Theory}.
\newblock Springer Publishing Company, Incorporated, 5th edition, 2017.
\newblock \href {https://doi.org/https://doi.org/10.1007/978-3-662-53622-3}
  {\path{doi:https://doi.org/10.1007/978-3-662-53622-3}}.

\bibitem{10.1145/2946799}
Michael Elberfeld, Martin Grohe, and Till Tantau.
\newblock Where first-order and monadic second-order logic coincide.
\newblock {\em {ACM} Trans. Comput. Log.}, 17(4):25, 2016.
\newblock \href {https://doi.org/10.1145/2946799} {\path{doi:10.1145/2946799}}.

\bibitem{elberfeld:hal-00678176}
Michael Elberfeld, Andreas Jakoby, and Till Tantau.
\newblock Algorithmic meta theorems for circuit classes of constant and
  logarithmic depth.
\newblock In Christoph D{\"{u}}rr and Thomas Wilke, editors, {\em 29th
  International Symposium on Theoretical Aspects of Computer Science, {STACS}
  2012, February 29th - March 3rd, 2012, Paris, France}, volume~14 of {\em
  LIPIcs}, pages 66--77. Schloss Dagstuhl - Leibniz-Zentrum f{\"{u}}r
  Informatik, 2012.
\newblock \href {https://doi.org/10.4230/LIPIcs.STACS.2012.66}
  {\path{doi:10.4230/LIPIcs.STACS.2012.66}}.

\bibitem{185392_parity_pos_determ}
E.~Allen Emerson and Charanjit~S. Jutla.
\newblock Tree automata, mu-calculus and determinacy (extended abstract).
\newblock In {\em 32nd Annual Symposium on Foundations of Computer Science, San
  Juan, Puerto Rico, 1-4 October 1991}, pages 368--377. {IEEE} Computer
  Society, 1991.
\newblock \href {https://doi.org/10.1109/SFCS.1991.185392}
  {\path{doi:10.1109/SFCS.1991.185392}}.

\bibitem{10.1007/978-3-642-31585-5_20}
John Fearnley and Sven Schewe.
\newblock Time and parallelizability results for parity games with bounded
  treewidth.
\newblock In Artur Czumaj, Kurt Mehlhorn, Andrew~M. Pitts, and Roger
  Wattenhofer, editors, {\em Automata, Languages, and Programming - 39th
  International Colloquium, {ICALP} 2012, Warwick, UK, July 9-13, 2012,
  Proceedings, Part {II}}, volume 7392 of {\em Lecture Notes in Computer
  Science}, pages 189--200. Springer, 2012.
\newblock \href {https://doi.org/10.1007/978-3-642-31585-5\_20}
  {\path{doi:10.1007/978-3-642-31585-5\_20}}.

\bibitem{ac0parity_sips}
Merrick~L. Furst, James~B. Saxe, and Michael Sipser.
\newblock Parity, circuits, and the polynomial-time hierarchy.
\newblock {\em Math. Syst. Theory}, 17(1):13--27, 1984.
\newblock \href {https://doi.org/10.1007/BF01744431}
  {\path{doi:10.1007/BF01744431}}.

\bibitem{DBLP:journals/corr/GajarskyLO13}
Jakub Gajarsk{\'{y}}, Michael Lampis, and Sebastian Ordyniak.
\newblock Parameterized algorithms for modular-width.
\newblock In Gregory~Z. Gutin and Stefan Szeider, editors, {\em Parameterized
  and Exact Computation - 8th International Symposium, {IPEC} 2013, Sophia
  Antipolis, France, September 4-6, 2013, Revised Selected Papers}, volume 8246
  of {\em Lecture Notes in Computer Science}, pages 163--176. Springer, 2013.
\newblock \href {https://doi.org/10.1007/978-3-319-03898-8\_15}
  {\path{doi:10.1007/978-3-319-03898-8\_15}}.

\bibitem{10.1007/978-3-662-46678-0_25}
Moses Ganardi.
\newblock Parity games of bounded tree- and clique-width.
\newblock In Andrew~M. Pitts, editor, {\em Foundations of Software Science and
  Computation Structures - 18th International Conference, FoSSaCS 2015, Held as
  Part of the European Joint Conferences on Theory and Practice of Software,
  {ETAPS} 2015, London, UK, April 11-18, 2015. Proceedings}, volume 9034 of
  {\em Lecture Notes in Computer Science}, pages 390--404. Springer, 2015.
\newblock \href {https://doi.org/10.1007/978-3-662-46678-0\_25}
  {\path{doi:10.1007/978-3-662-46678-0\_25}}.

\bibitem{GANIAN201488}
Robert Ganian, Petr Hlinen{\'{y}}, Joachim Kneis, Alexander Langer, Jan
  Obdrz{\'{a}}lek, and Peter Rossmanith.
\newblock Digraph width measures in parameterized algorithmics.
\newblock {\em Discret. Appl. Math.}, 168:88--107, 2014.
\newblock \href {https://doi.org/10.1016/j.dam.2013.10.038}
  {\path{doi:10.1016/j.dam.2013.10.038}}.

\bibitem{lmcs:5149}
Robert Ganian, Petr Hlinen{\'{y}}, Jaroslav Nesetril, Jan Obdrz{\'{a}}lek, and
  Patrice~Ossona de~Mendez.
\newblock Shrub-depth: Capturing height of dense graphs.
\newblock {\em Log. Methods Comput. Sci.}, 15(1), 2019.
\newblock \href {https://doi.org/10.23638/LMCS-15(1:7)2019}
  {\path{doi:10.23638/LMCS-15(1:7)2019}}.

\bibitem{shrub_intro}
Robert Ganian, Petr Hlinen{\'{y}}, Jaroslav Nesetril, Jan Obdrz{\'{a}}lek,
  Patrice~Ossona de~Mendez, and Reshma Ramadurai.
\newblock When trees grow low: Shrubs and fast {MSO1}.
\newblock In Branislav Rovan, Vladimiro Sassone, and Peter Widmayer, editors,
  {\em Mathematical Foundations of Computer Science 2012 - 37th International
  Symposium, {MFCS} 2012, Bratislava, Slovakia, August 27-31, 2012.
  Proceedings}, volume 7464 of {\em Lecture Notes in Computer Science}, pages
  419--430. Springer, 2012.
\newblock \href {https://doi.org/10.1007/978-3-642-32589-2\_38}
  {\path{doi:10.1007/978-3-642-32589-2\_38}}.

\bibitem{HUNTER2008206}
Paul Hunter and Stephan Kreutzer.
\newblock Digraph measures: Kelly decompositions, games, and orderings.
\newblock {\em Theor. Comput. Sci.}, 399(3):206--219, 2008.
\newblock \href {https://doi.org/10.1016/j.tcs.2008.02.038}
  {\path{doi:10.1016/j.tcs.2008.02.038}}.

\bibitem{8005092}
Marcin Jurdzinski and Ranko Lazic.
\newblock Succinct progress measures for solving parity games.
\newblock In {\em 32nd Annual {ACM/IEEE} Symposium on Logic in Computer
  Science, {LICS} 2017, Reykjavik, Iceland, June 20-23, 2017}, pages 1--9.
  {IEEE} Computer Society, 2017.
\newblock \href {https://doi.org/10.1109/LICS.2017.8005092}
  {\path{doi:10.1109/LICS.2017.8005092}}.

\bibitem{Lehtinen_2022}
Karoliina Lehtinen, Pawel Parys, Sven Schewe, and Dominik Wojtczak.
\newblock A recursive approach to solving parity games in quasipolynomial time.
\newblock {\em Log. Methods Comput. Sci.}, 18(1), 2022.
\newblock \href {https://doi.org/10.46298/lmcs-18(1:8)2022}
  {\path{doi:10.46298/lmcs-18(1:8)2022}}.

\bibitem{practical_synthesis}
Michael Luttenberger, Philipp~J. Meyer, and Salomon Sickert.
\newblock Practical synthesis of reactive systems from ltl specifications via
  parity games.
\newblock {\em Acta Informatica}, 57(1–2):3–36, Nov 2019.
\newblock URL: \url{http://dx.doi.org/10.1007/s00236-019-00349-3}, \href
  {https://doi.org/10.1007/s00236-019-00349-3}
  {\path{doi:10.1007/s00236-019-00349-3}}.

\bibitem{treedepth_def}
Jaroslav Nesetril and Patrice~Ossona de~Mendez.
\newblock Tree-depth, subgraph coloring and homomorphism bounds.
\newblock {\em Eur. J. Comb.}, 27(6):1022--1041, 2006.
\newblock \href {https://doi.org/10.1016/j.ejc.2005.01.010}
  {\path{doi:10.1016/j.ejc.2005.01.010}}.

\bibitem{sparsity}
Jaroslav Nesetril and Patrice~Ossona de~Mendez.
\newblock {\em Sparsity - Graphs, Structures, and Algorithms}, volume~28 of
  {\em Algorithms and combinatorics}.
\newblock Springer, 2012.
\newblock \href {https://doi.org/10.1007/978-3-642-27875-4}
  {\path{doi:10.1007/978-3-642-27875-4}}.

\bibitem{10.1007/978-3-540-45069-6_7}
Jan Obdrz{\'{a}}lek.
\newblock Fast mu-calculus model checking when tree-width is bounded.
\newblock In Warren A.~Hunt Jr. and Fabio Somenzi, editors, {\em Computer Aided
  Verification, 15th International Conference, {CAV} 2003, Boulder, CO, USA,
  July 8-12, 2003, Proceedings}, volume 2725 of {\em Lecture Notes in Computer
  Science}, pages 80--92. Springer, 2003.
\newblock \href {https://doi.org/10.1007/978-3-540-45069-6\_7}
  {\path{doi:10.1007/978-3-540-45069-6\_7}}.

\bibitem{10.5555/2392389.2392399}
Jan Obdrz{\'{a}}lek.
\newblock Clique-width and parity games.
\newblock In Jacques Duparc and Thomas~A. Henzinger, editors, {\em Computer
  Science Logic, 21st International Workshop, {CSL} 2007, 16th Annual
  Conference of the EACSL, Lausanne, Switzerland, September 11-15, 2007,
  Proceedings}, volume 4646 of {\em Lecture Notes in Computer Science}, pages
  54--68. Springer, 2007.
\newblock \href {https://doi.org/10.1007/978-3-540-74915-8\_8}
  {\path{doi:10.1007/978-3-540-74915-8\_8}}.

\end{thebibliography}

\appendix

\section{Proof of Lemma \ref{lemma:update_valid}}\label{sec:ski_lem_uev}
Let $\ParityGameFullDefC{\G'}{\V'}{\E'}{\R'}{\Color'}$ be a stop parity game
with player-aware coloring $\Color'$.
Let $\G''$ be a game created from $\G'$ by adding a directed edge from 
every vertex colored $\CFrom$ to every vertex colored $\CTo$.
For each $v \in \V'$ let $\EnfSetAlgG{\G}{\G'}{v}$ be a set of enforcements 
that is valid for $v$ in $\G'$.

\bigskip

Before proving \cref{lemma:update_valid} we will rewrite auxiliary gadgets 
$\GUpdateEV{\CFrom}{\CTo}$ and $\GUpdateOV{\CFrom}{\CTo}$ to a more convenient form.

\begin{definition}{Equivalent definition of $\GUpdateEV{\CFrom}{\CTo}$ }
  \begin{align*}
    \EnfOneMove &=& \bigcup_{w : \Color'(w) = \CTo}{\EnfSetAlgG{\G}{\G'}{w}}\\
    \EnfCycle &=&
    \{\FunAugment{P}{\CFrom}{\UndefinedM} : P(\CFrom) \text{ is even } \wedge P \in \EnfOneMove \}\\
    \EnfOpt &=& \EnfOneMove \cup \EnfCycle\\           
    \EnfSetAlgG{\G}{\G''}{v} &=&  \UpCloseECMP{\EnfSetAlgG{\G}{\G'}{v}} \cup
                          \UpCloseECMP{\{\EnfMerge{}{Q}{\CFrom}{R} : \CFrom \in \Domain{Q} \wedge Q \in \EnfSetAlgG{\G}{\G'}{v} \wedge R \in  \EnfOpt\}}
  \end{align*}
  Where $\EnfMerge{}{Q}{\CFrom}{R}$ is calculated as if $\CFrom$ was the color of player $\PE$,
  and $\UpCloseECMP{\EnfSetAlgG{\G}{\G'}{v}}$ means the up-closure of $\EnfSetAlgG{\G}{\G'}{v}$ according to $\ECMPEQ$.
\end{definition}

\begin{definition}{Equivalent definition of $\GUpdateOV{\CFrom}{\CTo}$ }
  \begin{align*}
    P \in \EnfSetAlgG{\G}{\G''}{v}\\
    \IFF\\
    \EX{Q \in \EnfSetAlgG{\G}{\G'}{v}}{\NoPairWith{Q}{\CFrom} \wedge Q \ECMPEQ P}\\
    \vee\\
    \EX{Q \in \EnfSetAlgG{\G}{\G'}{v}}{\FA{w: \Color'(w) = \CTo}{\EX{R \in \EnfSetAlgG{\G}{\G'}{w}}{\EnfMerge{\Color'}{Q}{\CFrom}{R} \ECMPEQ P}} \wedge \EvenPair{R}{\CFrom}}
  \end{align*}
  Where $\EnfMerge{}{Q}{\CFrom}{R}$ is calculated as if $\CFrom$ was the color of player $\PO$.
\end{definition}

Now we will prove \cref{lemma:update_valid} by proving soundness and completeness separately for  $\GUpdateEV{\CFrom}{\CTo}$ and $\GUpdateOV{\CFrom}{\CTo}$
(we have already observed that resulting sets are up-closed in \cref{sec:update_gadgets}).
\begin{lemma}\label{lemma:E_update_sound}
  Given that $\CFrom$ is the color of player $\PE$ and 
  for each $v \in \V'$ $\EnfSetAlgG{\G}{\G'}{v}$ is valid for $v$ in $\G'$,
  we have that for each $v \in \V'$ $\EnfSetAlgG{\G}{\G''}{v}$ created by $\GUpdateEV{\CFrom}{\CTo}$ is sound for $v$ in $\G''$.
\end{lemma}

\begin{proof}
    Let $v$ be an arbitrary vertex of $\G''$,
    we need to show that for every element $P$ of $\EnfSetAlgG{\G}{\G''}{v}$ there is
    a strategy $\ST_\PE$ that is safe from $v$ and $\EnforcementC{\G''}{\TreeModelName}{\ST_\PE}{v} \ECMPEQ P$.
    So let $P$ be any element of  $\EnfSetAlgG{\G}{\G''}{v}$,
    there are three cases to consider:
      \begin{itemize}
        \item If $P \in \EnfSetAlgG{\G}{\G'}{v}$,
        then from the assumption we have a strategy $\ST_\PE$
        that is safe from $v$ in $\G'$ and $\EnforcementC{\G'}{\TreeModelName}{\ST_\PE}{v} \ECMPEQ P$.
        This strategy is also safe from $v$ in $\G''$ and has the same enforcement, as it does not use new moves.
        \item If $\EnfMerge{\Color'}{Q}{\CFrom}{R} \ECMPEQ P$ where $Q \in \EnfSetAlgG{\G}{\G'}{v}$ and  $R \in \EnfOneMove$,
        then in fact $R \in \EnfSetAlgG{\G}{\G'}{w}$ for some $w$ such that $\Color'(w) = \CTo$.
        From the assumption we have a strategy $\ST_\PE^Q$ that is safe from $v$ in $\G'$ and 
        strategy $\ST_\PE^R$ that is safe from $w$ in $\G'$ such that 
        $Q' =\EnforcementC{\G'}{\TreeModelName}{\ST_\PE^Q}{v} \ECMPEQ Q$ and
        $R' = \EnforcementC{\G'}{\TreeModelName}{\ST_\PE^R}{w} \ECMPEQ R$.
        As $Q' \ECMPEQ Q$ and $R' \ECMPEQ R$, so by \cref{lemma:properties_of_merge} 
        we have $\EnfMerge{\Color'}{Q'}{\CFrom}{R'} \ECMPEQ  \EnfMerge{\Color'}{Q}{\CFrom}{R}$
        and that $\EnfMerge{\Color'}{Q'}{\CFrom}{R'}$ is the enforcement of some strategy in $\G''$ that is safe from $v$.
        \item If $\EnfMerge{\Color'}{Q}{\CFrom}{R} \ECMPEQ P$ where $Q \in \EnfSetAlgG{\G}{\G'}{v}$ and  $R \in \EnfCycle$,
      then by definition $R = \FunAugment{Z}{\CFrom}{\UndefinedM}$ where $Z(\CFrom)$ is even and  $Z \in \EnfSetAlgG{\G}{\G'}{w}$ for some $w$ such that $\Color'(w) = \CTo$.
      By assumption, we have $\ST_\PE^Z$ that is safe from $w$ and $Z' = \EnforcementC{\G'}{\TreeModelName}{\ST_\PE^Z}{w} \ECMPEQ Z$.
      By \cref{lemma:properties_of_merge}, $\FunAugment{Z'}{\CFrom}{\UndefinedM}$ is the enforcement from $w$ of a strategy in $\G''$
      that is safe from $w$.
      The rest of the proof is now analogical as above.
      \end{itemize}
  \end{proof}

\begin{lemma}\label{lemma:E_update_complete}
  Given that $\CFrom$ is the color of player $\PE$ and 
  for each $v \in \V'$ $\EnfSetAlgG{\G}{\G'}{v}$ is valid for $v$ in $\G'$,
  we have that for each $v \in \V'$ $\EnfSetAlgG{\G}{\G''}{v}$ created by $\GUpdateEV{\CFrom}{\CTo}$ is complete for $v$ in $\G''$.
\end{lemma}

To prove this, it suffices to show that for each
positional strategy $\ST_\PE$ that is safe from some vertex $v$ in $\G''$, we have
$\EnforcementC{\G''}{\TreeModelName}{\ST_\PE}{v} \in \EnfSetAlgG{\G}{\G''}{v}$.
Let $\ST_\PE$ be any positional strategy that is safe from a vertex $v$ in $\G''$, and $\ST_\PE^\STOP$ be its modification that behaves 
like $\ST_\PE$ but each time $\ST_\PE$ will make a move along a new edge (the edge that is present in $\G''$ but not in $\G'$)  $\ST_\PE^\STOP$ will 
decide to $\STOP$. 

\begin{observation}\label{observation:g_j_1_g_j}
  From the definition of $\ST_\PE^\STOP$, we have that
  its enforcements are the same in $\G'$ and $\G''$, 
  as it does not move along new edges.
\end{observation}

The main intuition behind the reasoning presented below is that, as 
we have edges from all vertices colored $\CFrom$ to all vertices colored $\CTo$,
so we should be able to modify $\ST_\PE$ such that if it moves along a new edge from $\CFrom$
to $\CTo$, then it moves always to the same vertex.
We should also be able to do the change without making enforcements from $v$ bigger according to $\ECMPEQ$.

To see this, we consider a sequence of worst arrivals (from the perspective of player $\PE$)
at vertices of color $\CFrom$, such that the next move according to $\ST_\PE$ is either a move along a new edge or $\STOP$.
More formally, we consider $\PlayPrefSpec_i, B_i, X_i, v_i$ defined as follows.\\
For $i = 0$:
\begin{itemize}
  \item $v_i = v$.
  \item $B_i$ is the set of plays $\Play$ that start at $v_i$ are consistent with $\ST_\PE^\STOP$
  and end at some vertex of color $\CFrom$ such that $\PlayMaxRankS{\G''}{\Play}$ is the lowest possible according to $\RCMPEQ$.

  More formally:\\
  Let $L = \PlayFrom{\G''}{\ST_\PE^\STOP}{v_i} \cap \bigcup_{w: \Color'(w) = \CFrom}{\PlayFromTo{\G''}{\ST_\PE^\STOP}{v_i}{w}}$,\\
  $r = \min_{\RCMPEQ}{\{\PlayMaxRankS{\G''}{\Play} : \Play \in L\}}$,\\ 
  then $B_i = \{\Play \in L : \PlayMaxRankS{\G''}{\Play} = r \}$.
  \item $X_i$ is the subset of $B_i$ that consists of plays $\Play$ such that $\ST_\PE(\Play) = \STOP$.
  \item $\PlayPrefSpec_i$ is any element of  $X_i$ and in case $X_i$ is empty, then it is any element of 
  $B_i$ and if $B_i$ is empty then so $\PlayPrefSpec_i$.
\end{itemize}
For $i > 0$:
\begin{itemize}
  \item If either $X_{i - 1} \not= \emptyset$ or $B_{i - 1} = \emptyset$, then $B_i = X_i = \emptyset$, $\IsUndefined{v_i}$ and $\PlayPrefSpec_i = \PlayPrefSpec_{i-1}$
  \item Otherwise, let $v_i = \ST_\PE(\PlayPrefSpec_{i-1})$, and 
  $B_i, X_i$ be defined exactly as when $i = 0$. Let $\PlayPrefSpec_i$ be $\PlayPrefSpec_{i - 1}$  extended with any element of 
  $X_i$ and in case $X_i = \emptyset$, then any element of $B_i$ and if $B_i = \emptyset$, then let $\PlayPrefSpec_i$ be equal to $\PlayPrefSpec_{i - 1}$.
  Note that there for $i > 0$ we have that $\Color'(\PlayLastVertex{\PlayPrefSpec_{i - 1}}) = \CFrom$ and $\Color'(v_i) = \CTo$.
\end{itemize}
We denote $\PlayPrefSpec$ as $\lim_{i \Approaches \infty}{\PlayPrefSpec_i}$, and we consider cases for $\PlayPrefSpec$. 
As $\PlayInfSpec_i$, we denote the part appended to $\PlayPrefSpec_{i - 1}$ to create $\PlayPrefSpec_i$.

\begin{observation}\label{observation:compliance_1}
  If $\PlayPrefSpec$ is not empty, then it starts at $v$ and is either consistent with $\ST_\PE$ or can be extended to a play that is consistent with $\ST_\PE$.
  What is more,
  each $\PlayPrefSpec_i$ can be extended to a play that is consistent with $\ST_\PE$.
\end{observation}

First of all, if $\PlayPrefSpec$ is empty, then $B_0 = \emptyset$ and
every play that starts at $v$ and is consistent with $\ST_\PE$ never makes use of new moves, so 
$\ST_\PE$ is also safe from $v$ in $\G'$ and has the same enforcement, so we have 
$\EnforcementC{\G''}{\TreeModelName}{\ST_\PE}{v} \in \EnfSetAlgG{\G}{\G''}{v}$.

\bigskip

If $\PlayPrefSpec$ is not empty, then there may exist $i > 0$ such that $B_{i} \not= \emptyset$ and 
$\PlayInfSpec_i$ has even max $\R$. 
Now assume that such $i$ exists and let it be the smallest possible. 
Consider the strategy 
$\ST'_\PE$ that behaves like  ${\ST}_\PE$ but each time   ${\ST}_\PE$ would use a new move or decide 
to $\STOP$ at vertex colored $\CFrom$ the strategy
$\ST'_\PE$ makes a move to $v_i$.
Note that 
that $\ST'_\PE$ defined as above is positional,
as ${\ST}_\PE$ is positional.

\begin{observation}\label{observation:till_i_not_better}
  For $i$ as above, we have
  $\PlayMaxRankS{\G''}{\PlayPrefSpec_{i - 1}} \RCMPEQ  \PlayMaxRankS{\G''}{\PlayPrefSpec_{0}}$,
  as otherwise
  contradiction with $i$ being minimal.
\end{observation}

\begin{lemma}\label{lemma:merge_composition} 
  \phantom{L}\\
  For
  \begin{gather*}
    P = \EnforcementC{\G'}{\TreeModelName}{\ST_\PE^\STOP}{v}\\
    Q = \EnfLoop{\EnforcementC{\G'}{\TreeModelName}{\ST_\PE^\STOP}{v_i}}{\CFrom}
  \end{gather*} 
  we have
  \begin{gather*}
    \EnforcementC{\G''}{\TreeModelName}{\ST'_\PE}{v} = \EnfMerge{\Color'}{P}{\CFrom}{Q}
  \end{gather*}
  what is more $\ST'_\PE$ is safe from $v$ and $\EnforcementC{\G''}{\TreeModelName}{\ST'_\PE}{v} \ECMPEQ \EnforcementC{\G''}{\TreeModelName}{{\ST}_\PE}{v}$.
\end{lemma}
\begin{proof}
  From the definition of $\ST_\PE^\STOP$ and \cref{observation:compliance_1}  we have that $\ST_\PE^\STOP$ is safe from both $v$ and $v_i$ in $\G''$.
  From the choice of $i$, we have that $\EnforcementC{\G'}{\TreeModelName}{\ST_\PE^\STOP}{v_i}(\CFrom)$ is even.
  This along with  \cref{lemma:properties_of_merge} gives us that $Q$ is the enforcement from $v_i$ of
  a strategy ${\ST}_\PE^{\STOP'}$ that is safe from $v_i$. 

  From the same lemma, we have that  $\EnfMerge{\Color'}{P}{\CFrom}{Q}$ 
  is the enforcement from $v$ of the strategy  $\ST^{\prime\prime}_\PE$  that is safe from $v$, behaves like ${\ST}_\PE^\STOP$,
  but each time ${\ST}_\PE^\STOP$ decides to stop at some vertex colored $\CFrom$
  $\ST^{\prime\prime}_\PE$ decides to move to $v_i$, so it is also the enforcement from $v$ of  $\ST'_\PE$.
  
  For the relation $\EnforcementC{\G''}{\TreeModelName}{\ST'_\PE}{v} \ECMPEQ \EnforcementC{\G''}{\TreeModelName}{{\ST}_\PE}{v}$,
  let $\PlayPref \in \left(\PlayFromTo{\G''}{\ST'_\PE}{v}{w} \setminus \PlayFromTo{\G''}{{\ST}_\PE}{v}{w}\right)$  and $q$ be the last point on $\PlayPref$ such that
  $\ST'_\PE(\PlayPref[q]) \not= \ST_\PE(\PlayPref[q])$ and $q$ is not the last point on $\PlayPref$, then $\Color'(\PlayPref[q]) = \CFrom$ and we can replace 
  $\PlayPref[\ldots q]$ (prefix consisting of the first $q$ elements) with $\PlayPrefSpec_{i-1}$ creating $\PlayPref'$, such that $\PlayMaxRank{\G''}(\PlayPref') \RCMPEQ \PlayMaxRank{\G''}(\PlayPref)$
  and $\PlayPref' \in  \PlayFromTo{\G''}{{\ST}_\PE}{v}{w}$.
  This follows from \cref{observation:till_i_not_better} and the fact that in any play consistent with $\ST'_\PE$ fragments between vertices of color $\CFrom$
  at which  ${\ST}_\PE^\STOP$ decides to stop have rank not smaller that $\PlayMaxRank{\G''}(\PlayInfSpec_i)$ according to $\RCMPEQ$, 
  what in particular means even max rank.
\end{proof}

Note that in the lemma above $\EnforcementC{\G''}{\TreeModelName}{\ST'_\PE}{v} \in \EnfSetAlgG{\G}{\G''}{v}$ as 
$\EnforcementC{\G'}{\TreeModelName}{\ST_\PE^\STOP}{v} \in \EnfSetAlgG{\G}{\G'}{v}$
and
$\EnforcementC{\G'}{\TreeModelName}{\ST_\PE^\STOP}{v_i} \in \EnfSetAlgG{\G}{\G'}{v_i}$.
So as $\EnforcementC{\G''}{\TreeModelName}{\ST'_\PE}{v} \ECMPEQ \EnforcementC{\G''}{\TreeModelName}{{\ST}_\PE}{v}$ and $\EnfSetAlgG{\G}{\G''}{v}$ is up-closed, we have that $\EnforcementC{\G''}{\TreeModelName}{{\ST}_\PE}{v} \in \EnfSetAlgG{\G}{\G''}{v}$.

\bigskip

The one last case is when $\PlayPrefSpec$  is not empty but there is no $i > 0$ such that 
$B_{i} \not= \emptyset$ and 
$\PlayInfSpec_i$ has even max $\R$. As $\ST_\PE$ is a positional strategy that is safe from $v$, 
so from \cref{observation:compliance_1} $\PlayPrefSpec$ must be finite as each of 
$\PlayInfSpec_i$ for $i > 0$ is either empty or has odd max rank. 
Now let $u$ be the last vertex on $\PlayPrefSpec$ and $i$ be the highest number such that $B_i \not= \emptyset$, observe that $\Color'(u) = \CFrom$.
Consider the strategy $\ST'_\PE$ that behaves exactly like ${\ST}_\PE$, but for vertices $w$ of color $\CFrom$ such that
${\ST}_\PE^{\STOP}(w) = \STOP$, we have
$\ST'_\PE(w) = {\ST}_\PE(u)$, i.e. each time ${\ST}_\PE$ chooses to move along new edge or decides to $\STOP$
at vertex colored $\CFrom$
the strategy $\ST'_\PE$ makes the move that ${\ST}_\PE$ would make at vertex $u$.
Note 
that $\ST'_\PE$ is positional,
as ${\ST}_\PE$ is positional.

\begin{observation}\label{observation:till_i_not_better_2}
  For $i$ as above, we have
  $\PlayMaxRankS{\G''}{\PlayPrefSpec_{i}} \RCMPEQ  \PlayMaxRankS{\G''}{\PlayPrefSpec_{0}}$
  as otherwise, we would have $k$ such that $\PlayInfSpec_k$ is even.
\end{observation}

\begin{lemma}\label{lemma:no_going_back}
  If $\ST_\PE(u) \not= \STOP$, then $\EnforcementC{\G''}{\TreeModelName}{\ST_\PE^\STOP}{\ST_\PE(u)}(\CFrom)$
  is \UndefinedT.
\end{lemma}
\begin{proof}
  If $\ST_\PE(u) \not= \STOP$, then $X_i = \emptyset$, and as $B_{i + 1} = \emptyset$, so there is no 
    play that starts at $\ST_\PE(u)$ is consistent with $\ST_\PE^\STOP$ and ends at vertex colored $\CFrom$.
\end{proof}

\begin{lemma}\label{lemma:merge_composition_2}
  \phantom{G}\\
  For
  \begin{gather*}
    P = \EnforcementC{\G'}{\TreeModelName}{\ST_\PE^\STOP}{v}\\
    w = \begin{cases}
      \ST_\PE(u) & \MIF \ST_\PE(u) \not= \STOP\\
      u       & \MOTHERWISE
    \end{cases}\\
    Q = \EnforcementC{\G'}{\TreeModelName}{\ST_\PE^\STOP}{w}
  \end{gather*} 
  we have
  \begin{gather*}
    \EnforcementC{\G''}{\TreeModelName}{\ST'_\PE}{v} = \EnfMerge{\Color'}{P}{\CFrom}{Q}
  \end{gather*}
  and $\ST'_\PE$ is safe from $v$ and $\EnforcementC{\G''}{\TreeModelName}{\ST'_\PE}{v} \ECMPEQ \EnforcementC{\G''}{\TreeModelName}{{\ST}_\PE}{v}$.
\end{lemma}

\begin{proof}
  From the definition of ${\ST}_\PE^\STOP$ and \cref{observation:compliance_1}
  we have that ${\ST}_\PE^\STOP$ is safe from $v$ and $w$ in $\G''$.
  From \cref{lemma:properties_of_merge},
  we have that  $\EnfMerge{\Color'}{P}{\CFrom}{Q}$ 
  is the enforcement from $v$ of the strategy  $\ST^{\prime\prime}_\PE$  that is safe from $v$, behaves like ${\ST}_\PE^\STOP$,
  but the first time ${\ST}_\PE^\STOP$ decides to stop at some vertex colored $\CFrom$
  $\ST^{\prime\prime}_\PE$ makes the move that $\ST_\PE$ would make at $u$. 
  Note that from \cref{lemma:no_going_back} this means either stopping the game or playing further according to ${\ST}_\PE^\STOP$
  and never stopping at any vertex of color $s$ again, so in fact $\EnfMerge{\Color'}{P}{\CFrom}{Q}$ is the enforcement of $\ST_\PE'$.

  For the relation $\EnforcementC{\G''}{\TreeModelName}{\ST'_\PE}{v} \ECMPEQ \EnforcementC{\G''}{\TreeModelName}{{\ST}_\PE}{v}$,
  let $\PlayPref \in \left(\PlayFromTo{\G''}{\ST'_\PE}{v}{w} \setminus \PlayFromTo{\G''}{{\ST}_\PE}{v}{w}\right)$  and $q$ be the first point on $\PlayPref$ such that
  $\ST'_\PE(\PlayPref[q]) \not= \ST_\PE(\PlayPref[q])$, then $\Color'(\PlayPref[q]) = \CFrom$, and we can replace 
  $\PlayPref[\ldots q]$ with $\PlayPrefSpec_{i}$ creating $\PlayPref'$, such that $\PlayMaxRank{\G''}(\PlayPref') \RCMPEQ \PlayMaxRank{\G''}(\PlayPref)$
  and $\PlayPref' \in  \PlayFromTo{\G''}{{\ST}_\PE}{v}{w}$. 
  This follows from \cref{observation:till_i_not_better_2} and 
  \cref{lemma:no_going_back}, as for $q' < q$ we have ${\ST}_\PE^\STOP(\PlayPref[q']) \not= \STOP$
  (at $u$ we either stop or move to a vertex from which we will play according to $\ST_\PE$ and never visit a vertex of color $\CFrom$ again).
  
\end{proof}

Note that in the lemma above $\EnforcementC{\G''}{\TreeModelName}{\ST'_\PE}{v} \in \EnfSetAlgG{\G}{\G''}{v}$,
so as $\EnfSetAlgG{\G}{\G''}{v}$ is up-closed we have that $\EnforcementC{\G''}{\TreeModelName}{{\ST}_\PE}{v} \in \EnfSetAlgG{\G}{\G''}{v}$.
This last case finishes the proof of \cref{lemma:E_update_complete}.

\bigskip

Now lets focus on  $\GUpdateOV{\CFrom}{\CTo}$.
\begin{lemma}\label{lemma:O_update_sound}
  Given that $\CFrom$ is the color of player $\PO$ and 
  for each $v \in \V'$ $\EnfSetAlgG{\G}{\G'}{v}$ is valid for $v$ in $\G'$,
  we have that for each $v \in \V'$ $\EnfSetAlgG{\G}{\G''}{v}$ created by $\GUpdateOV{\CFrom}{\CTo}$ is sound for $v$ in $\G''$.
  \end{lemma}
\begin{proof}
    Let $v$ be an arbitrary vertex of $\G''$,
    we need to show that for every $P \in \EnfSetAlgG{\G}{\G''}{v}$ there exists a strategy $\ST_\PE$
    that is safe from $v$ and $\EnforcementC{\G''}{\TreeModelName}{\ST_\PE}{v} \ECMPEQ P$.
    So let $P$ be the arbitrary element of $\EnfSetAlgG{\G}{\G''}{v}$.

    If $\EX{Q \in \EnfSetAlgG{\G}{\G'}{v}}{\NoPairWith{Q}{\CFrom} \wedge Q \ECMPEQ P}$, then
    as $\EnfSetAlgG{\G}{\G'}{v}$ is valid for $v$ in $\G'$, so we have a strategy $\ST_\PE$ such that
    $\EnforcementC{\G'}{\TreeModelName}{\ST_\PE}{v} \ECMPEQ Q$, the strategy $\ST_\PE$
     is safe from 
    $v$ in $\G'$ and there is no play consistent with $\ST_\PE$ that starts at $v$ and visits
    a vertex colored $\CFrom$. So $\ST_\PE$ is safe from $v$ in $\G''$ with the same enforcement.
    
    For the other case, let $w_1, \ldots, w_m$ be all vertices from the arena $\G''$ that have color $\CTo$.
    Observe that from the definition of the update rule, we have 
    $Q \in \EnfSetAlgG{\G}{\G'}{v}$, and for each $w_i$ we have $R_i \in \EnfSetAlgG{\G}{\G'}{w_i}$
    such that $\EnfMerge{\Color'}{Q}{\CFrom}{R_i} \ECMPEQ P$ and $\EvenPair{R_{i}}{\CFrom}$. By \cref{lemma:properties_of_merge},
    we can choose $Q$ and $R_i$ to be minimal (i.e $\FA{Q' \in \EnfSetAlgG{\G}{\G'}{v}}{Q' \not \ECMP Q}$ 
    and $\FA{R'_i \in \EnfSetAlgG{\G}{\G'}{w_i}}{R'_i \not \ECMP R_i}$),
    and we will still have the property that $\EnfMerge{\Color'}{Q}{\CFrom}{R_i} \ECMPEQ P$ and $\EvenPair{R_{i}}{\CFrom}$. 
    For such choice, by
    assumed validity, we have strategies 
    $\ST_\PE^v$ safe from $v$, and $\ST_\PE^{w_i}$ safe from $w_i$ that have exactly those enforcements in $\G'$. Using those 
    strategies we construct a strategy $\ST_\PE$ that initially behaves exactly like $\ST_\PE^v$,
    and each time player $\PO$ makes a move along new edge to the vertex $w_i$  $\ST_\PE$
    forgets about the past and switches to $\ST_\PE^{w_i}$.
    Observe that on a  play starting at $v$ and consistent with  $\ST_\PE$ 
    parts between switching the strategies have the highest rank even as we have $\EvenPair{R_i}{\CFrom}$.
    Observe that this along with $\EnfMerge{\Color'}{Q}{\CFrom}{R_i} \ECMPEQ P$ means that $\EnforcementC{\G''}{\TreeModelName}{\ST_\PE}{v} \ECMPEQ P$ and
    that $\ST_\PE$ is safe from $v$ as if the switch is performed infinitely often, then 
    the play is winning for player $\PE$. Otherwise, the safety follows from the safety of 
    the last strategy that we 
     switched to.

 \end{proof}

\begin{lemma}\label{lemma:O_update_complete}
  Given that $\CFrom$ is the color of player $\PO$ and 
  for each $v \in \V'$ $\EnfSetAlgG{\G}{\G'}{v}$ is valid for $v$ in $\G'$,
  we have that for each $v \in \V'$ $\EnfSetAlgG{\G}{\G''}{v}$ created by $\GUpdateOV{\CFrom}{\CTo}$ is complete for $v$ in $\G''$.
\end{lemma}
\begin{proof}
    We need to show that for any positional strategy $\ST_\PE$ of player $\PE$ that is save from some vertex $v$
    in $\G''$, we have $\EnforcementC{\G''}{\TreeModelName}{\ST_\PE}{v} \in \EnfSetAlgG{\G}{\G''}{v}$.
    Let $\ST_\PE$ be any positional strategy that is safe from $v$ and let $P =  \EnforcementC{\G''}{\TreeModelName}{\ST_\PE}{v}$.

    First of all if $\NoPairWith{P}{\CFrom}$, then  $P \in \EnfSetAlgG{\G}{\G'}{v}$
    and $P \in \EnfSetAlgG{\G}{\G''}{v}$.

    Otherwise let $w_1, \ldots, w_m$ be all vertices from the arena $\G''$ that have color $\CTo$.
    As now there is a play that starts at $v$ is consistent with $\ST_\PE$ and visits a vertex 
    colored $\CFrom$, so $\ST_\PE$ is safe from each of the vertices $w_i$ 
    (also note that $\ST_\PE$ is safe from $v$ and each of the vertices $w_i$ in $\G'$).
    Consider
     $Q = \EnforcementC{\G'}{\TreeModelName}{\ST_\PE}{v}$,
    $R_i = \EnforcementC{\G'}{\TreeModelName}{\ST_\PE}{w_i}$
    and observe that $\EvenPair{R_i}{\CFrom}$ must hold as otherwise player $\PO$
    would able to close cycles with the highest rank being odd.  What is more 
    from \cref{lemma:properties_of_merge} we have that $\EnfMerge{\Color'}{Q}{\CFrom}{R_i} \ECMPEQ P$,
    so $Q$ and $R_i$ fit to the condition for adding $P$ to $\EnfSetAlgG{\G}{\G''}{v}$.
\end{proof}

Now proof of \cref{lemma:update_valid} follows easily from \cref{lemma:E_update_sound,lemma:E_update_complete,lemma:O_update_sound,lemma:O_update_complete}.

\end{document}